\documentclass[preprint,12pt]{elsarticle}

\usepackage{geometry}
\usepackage{float}
\usepackage{mathrsfs}
\usepackage{amsmath}
\usepackage{amssymb}
\usepackage{graphicx}
\usepackage{subfigure}
\usepackage{color} % For colour coding during editing - remove in final

\definecolor{darkred}{rgb}{.7,0,0}

\definecolor{darkgreen}{rgb}{0,0.7,0}

\definecolor{darkblue}{rgb}{0,0,0.7}

\usepackage{enumerate}
\floatstyle{ruled}
\newfloat{algorithm}{tbp}{loa}
\providecommand{\algorithmname}{Algorithm}
\floatname{algorithm}{\protect\algorithmname}

\usepackage{amsthm}

\usepackage{mathrsfs}

\usepackage{amsfonts}

\usepackage{epsfig}

\usepackage{bm}

\usepackage{mathrsfs}

\usepackage{subfig}\usepackage[all]{xy}

\newcommand{\hu}{\widehat {u}}

\newcommand{\bbE}{\mathbb{E}}

\newcommand{\bbR}{\mathbb{R}}

\newcommand{\cG}{\mathcal{G}}
\newcommand{\cO}{\mathcal{O}}

\newcommand{\testfunction}{\ensuremath{\mathcal{B}_b}}

\newtheorem{theorem}{Theorem}[section]
\newtheorem{lem}{Lemma}[section]

\newtheorem{prop}{Proposition}[section]

\newcounter{hypA}
\newenvironment{hypA}{\refstepcounter{hypA}\begin{itemize}
  \item[({\bf A\arabic{hypA}})]}{\end{itemize}}

\begin{document}

\begin{frontmatter}



\title{Multilevel Sequential Monte Carlo Samplers}


\author{Alexandros Besko$\textrm{s}^1$,  Ajay Jasr$\textrm{a}^2$, Kody La$\textrm{w}^3$, Raul Tempon$\textrm{e}^3$ \& Yan Zho$\textrm{u}^{2}$}

\address{$^{1}$Department of Statistical Science,
University College London, London, WC1E 6BT, UK,\\ $^{2}$Department of Statistics \& Applied Probability,
National University of Singapore, Singapore, 117546, SG,\\ \&  $^{3}$Center for Uncertainty Quantification
in Computational Science \& Engineering, King Abdullah University of Science and Technology, Thuwal, 23955-6900, KSA.}

\begin{abstract}
%We consider t
In this article we consider the approximation of expectations w.r.t.~probability distributions associated to the solution of partial differential equations (PDEs); this scenario %such problems
appears routinely in Bayesian inverse problems. In practice, one often has to solve the associated PDE numerically, using, for instance finite element methods
and leading to a discretisation bias, with the step-size level $h_L$. In addition, the expectation cannot be computed analytically and one often resorts to Monte Carlo methods.
In the context of this problem, it is known that the introduction of the multilevel Monte Carlo (MLMC) method can reduce the amount of
computational effort to estimate expectations, for a given level of error. This is achieved via a telescoping identity associated to a Monte Carlo approximation of a sequence of probability
distributions with discretisation levels $\infty>h_0>h_1\cdots>h_L$. In many practical problems of interest, one cannot achieve an i.i.d.~sampling of the associated sequence of
probability distributions. %We introduce a 
A sequential Monte Carlo (SMC) version of the MLMC method is introduced to deal with this problem. 
%We show that, 
It is shown that under appropriate assumptions, the attractive  
property of a reduction of the amount of computational effort to estimate expectations, for a given level of error, can be maintained within the SMC context. %Our
The approach is numerically illustrated on a Bayesian inverse problem.
\end{abstract}



 

\end{frontmatter}

%+Title
%\begin{center}
%
%{\Large \textbf{multilevel Sequential Monte Carlo Samplers}}
%
%\vspace{0.5cm}
%
%BY ALEXANDROS BESKOS, KODY J. H. LAW, AJAY JASRA \& RAUL TEMPONE
%
%{\footnotesize Department of Statistical Science,
%University College London, London, WC1E 6BT, UK.}
%{\footnotesize E-Mail:\,}\texttt{\emph{\footnotesize a.beskos@ucl.ac.uk}}\\
%{\footnotesize Department of Statistics \& Applied Probability,
%National University of Singapore, Singapore, 117546, SG.}
%{\footnotesize E-Mail:\,}\texttt{\emph{\footnotesize staja@nus.edu.sg}}\\
%{\footnotesize Center for Uncertainty Quantification
%in Computational Science \& Engineering, King Abdullah University of Science and Technology, Thuwal, 23955-6900, KSA.}
%{\footnotesize E-Mail:\,}\texttt{\emph{\footnotesize  kody.law@kaust.edu.sa; raul.tempone@kaust.edu.sa}}
%\end{center}

\section{Introduction}\label{sec:intro}

%We c
Consider a sequence of probability measures $\{\eta_{l}\}_{l\geq 0}$ on a common measurable 
space $(E,\mathcal{E})$; we assume that the probabilities have common dominating
finite-measure $du$ and write the densities w.r.t.~$du$ as $\eta_l=\eta_l(u)$. In particular, for some known $\gamma_{l}:E\rightarrow\mathbb{R}^+$, we let
\begin{equation}
\label{eq:target}
\eta_{l}(u) = \frac{\gamma_l(u)}{Z_l}
\end{equation}
where the normalizing constant $Z_l = \int_E\gamma_l(u)du$ may be unknown.  The context of interest is when the sequence of densities is
associated to an `accuracy' parameter $h_l$, with $h_l\rightarrow 0$ 
as $l\rightarrow \infty$ with $\infty>h_0>h_1>\cdots>h_{\infty}=0$. 
This set-up is relevant to the context of
discretised  numerical approximations of continuum fields, as we will explain below. 
%We are interested 
The objective is to compute:
$$
\mathbb{E}_{\eta_\infty}[g(U)] := \int_E g(u)\eta_\infty(u)du 
$$
for potentially many measurable $\eta_\infty-$integrable functions $g:E\rightarrow\mathbb{R}$. In practice one cannot treat $h_l=0$ and
%as it will be associated to a numerical approximation  
 must consider these distributions with $h_l>0$.

Problems involving numerical approximations of continuum fields are discretized before being solved numerically.  Finer-resolution solutions are more expensive to compute than coarser ones.
Such discretizations naturally give rise to hierarchies of resolutions via the use of nested meshes.
Successive solution at refined meshes can be utilized to mitigate the number of necessary solves for the 
finest resolutions.  For the solution of linear systems, the coarsened systems are solved as pre-conditioners 
within the framework of iterative linear solvers in order to reduce the condition number, and hence the 
number of necessary iterations at the finer resolution.  This is the principle of multi-grid methods.
For Monte Carlo methods, as in the context above, a telescoping sum of associated differences at successive refinement levels 
can be utilized. This is so that the bias of the resulting multilevel estimator is determined by the finest 
level but the variance of the estimators of the differences decays.  
The reduction in the variance at finer levels implies that the number of samples 
%(and hence cost) 
required to reach a given error tolerance is also reduced with increasing resolution.  This procedure  is then optimized to balance 
the extra per-sample cost at the finer levels.
Overall one can obtain a method with smaller computational effort to reach a pre-determined error than applying a standard Monte Carlo method immediately at the finest resolution \cite{gile:08}.

Multi Level Monte Carlo (MLMC) \cite{gile:08} (see also \cite{hein:98}) methods are such that 
one typically sets an error threshold for a target expectation, 
and then sets out to attain an estimator with the prescribed error 
utilizing an optimal allocation of Monte Carlo resources.  
Within the context of \cite{gile:08, hoan:12}, the continuum problem is a
stochastic differential equation (SDE) or PDE with random coefficients, and the target quantity is an expectation of a functional, say $g: E \rightarrow \bbR$,  
of the parameter of interest $U\in E$, over an ideal measure $U\sim \eta_{\infty}$ 
that avoids discretisation.  The levels are a hierarchy of refined approximations 
of the function-space, specified  in terms of a small resolution parameter say $h_l$, 
for $0\le l \le L$, thus giving rise to a corresponding sequence of approximate laws 
$\eta_l$.
The method uses the telescopic sum
$$
\bbE_{\eta_L}[g(U)] = \bbE_{\eta_0}[g(U)] + \sum_{l=1}^L 
\{\bbE_{\eta_l}[g(U)]-\bbE_{\eta_{l-1}}[g(U)]\}
$$
and proceeds by coupling the consecutive probability 
distributions $\eta_{l-1}$, $\eta_{l}$.  
Thus, the expectations are estimated via the standard unbiased Monte Carlo 
averages $$Y^{N_l}_l =  \sum_{i=1}^{N_l} \{ g(U_l^{(i)})-g(U_{l-1}^{(i)})\}N_l^{-1}$$
where $\{U_{l-1}^{(i)},U_l^{(i)}\}$ are i.i.d.\@ samples,
with marginal laws $\eta_{l-1}$, $\eta_l$, respectively, carefully constructed on a joint
probability space. 
This is repeated 
independently for $0\le l\le  L$.
The overall multilevel estimator will be 
\begin{equation}
\label{eq:multi}
\hat{Y}_{L,{\rm Multi}} = \sum_{l=0}^{L} Y^{N_l}_l\ ,
\end{equation}
under the convention that $g(U_{-1}^{(i)})=0$.
A simple error analysis gives that the mean squared error (MSE) is
\begin{equation}
%MSE(\hat{g}^{L,{\rm Multi}}) = 
\bbE \{ \hat{Y}_{L,{\rm Multi}} - \bbE_{\eta_{\infty}} [g(U)] \}^2 =  \underbrace{\bbE
\{ \hat{Y}_{L,{\rm Multi}}- \bbE_{\eta_{L}} [{g}(U)]\}^2}_{\rm variance}  
+ \underbrace{\{\bbE_{\eta_L} [{g}(U)] - \bbE_{\eta_{\infty}} [g(U)]\}^2}_{\rm bias}\  .
\label{eq:mse}
\end{equation}
One can now optimally allocate $N_0, N_1,\ldots, N_L$ to minimize the variance term  
 $\sum_{l=0}^L V_l/N_l$ for fixed  
computational cost $\sum_{l=0}^L C_l N_l$,
where $V_l$ is the variance of  $[g(U_l^{(i)})-g(U_{l-1}^{(i)})]$ and  $C_l$
the computational cost for its realisation. Using Lagrange multipliers for the above constrained optimisation, we get the optimal allocation of resources 
$N_l \propto \sqrt{V_l/C_l}$.
In more detail, 
 %
%one can simply replace the expectations with their unbiased estimators and yield an estimator $\hat{g}_{L,{\rm Multi}}$
%with bias $\cO(h_L^\alpha)$, variance given by $\cO(\sum_{l=0}^L V_l/N_l)$, and cost given by $\cO(\sum_{l=0}^L C_l N_l)$,
%where $V_l, C_l$ are the variance and computational cost of the $l^{th}$ increment  This naturally fits into the variance/bias decomposition of the mean square error
%\begin{equation}
%%MSE(\hat{g}^{L,{\rm Multi}}) = 
%\bbE [ \hat{g}_{L,{\rm Multi}} - \bbE (g) ]^2 =  \underbrace{\bbE [ \hat{g}_{L,{\rm Multi}}- \bbE ({g}_{L})]^2}_{\rm variance}  
%+ \underbrace{[\bbE ({g}_{L}) - \bbE (g)]^2}_{\rm bias}
%\label{eq:mse}
%\end{equation}
%since we know the mean $\bbE (\hat{g}_{L,{\rm Multi}}) = \bbE ({g}_{L}) $ and its bias.
%So, roughly speaking one can simply minimize the cost given the constraint that the variance matches the 
%bias, $\sum_{l=0}^L V_l/N_l = \cO(h_L^{2\alpha})$
%(one can be more precise, including exact constants, but this is the principle).
%
the typical chronology is that  one targets an MSE, say $\mathcal{O}(\epsilon^2)$, then
 (i) given a characterisation of the bias as an order of $h_l$,
 one 
 determines $h_l=M^{-l}$, $l=0,1,\ldots,L$, for some integer $M>1$,
 and 
  chooses a horizon $L$ such that the bias is  
$\cO(\epsilon^2)$ and
(ii) given a characterisation of $V_l$, $C_l$ as some orders of $h_l$, one
optimizes the required samples $N_0,\ldots N_L$ needed to give variance $\cO(\epsilon^2)$. 
Thus, a specification of the bias, variance and computational
costs as  functions of $h_l$ is needed. 

As a prototypical  example of the above setting \cite{gile:08}, consider the case $U=X(T)$ with $X(T)$ being the terminal position of the solution $X$ of a
 SDE
and $\eta_l$ is the distribution of $X(T)$ under the consideration of a numerical approximation with time-step $\Delta t_l = h_l$.  
%Then for the Euler-Marayuma method $\alpha=1$ and $\zeta=1$ \cite{}. 
The laws $\eta_{l-1}$, $\eta_l$ can be coupled via use of the same driving Brownian 
path.
Invoking the relevant error analysis 
for SDE models, one can obtain (for $U\sim \eta_{\infty}$, $U_{l}\sim \eta_l$,
and defined on the common probability space):
\begin{itemize}
\item[(i)] weak error $|\bbE [g(U_l)- g(U)]| =\cO(h_l^\alpha)$, providing the bias 
$\cO(h_l^\alpha)$, 
\item[(ii)] strong error,  $\bbE |g(U_l) - g(U)|^2 = \cO(h_l^\beta)$, giving the variance 
$V_l = \cO(h_l^\beta)$,
\item[(iii)] computational cost for a realisation of $g(U_l)-g(U_{l-1})$,  $C_l =\cO(h_l^{-\zeta})$,
\end{itemize}
for some constants $\alpha, \beta, \zeta$ related to the details of the
discretisation method.
%\begin{align}
%\nonumber
%(i){\rm ~~weak~~ error} & |\bbE (g_l - g)| &=\cO(h_l^\alpha), \\ 
%\nonumber
%(ii)& {\rm ~~storng~~ error} & \bbE |g_l - g|^2 &= \cO(h_l), \\ 
%(iii)& {\rm ~~cost} & {\rm Cost}(g_l) & = \cO(h_l^{-\zeta}).
%\label{eq:mlmc_est}
%\end{align}
%Notice that one of the three can be normalized to $1$. 
 The standard Euler Marayuma method for solution of SDE %error analysis 
gives the orders  $\alpha=\beta=\zeta=1$. 

Assuming  a general context, 
given such rates for bias, $V_l$ and $C_l$, 
one proceeds as follows.
Recall that $h_l=M^{-(l+k)}$, for some integer $M>1$. 
%, and consider
%the Euler-Marayuma method, in which $\alpha=\beta=\zeta=1$ \cite{kloeden}.    
Then, targeting an error tolerance of $\epsilon$ and letting $h_L^{\alpha} =  M^{-L\alpha}=\mathcal{O}(\epsilon)$, one has $L=\log(\epsilon^{-1})/(\alpha \log(M)) + \cO(1)$, as in \cite{gile:08}.
%The cost will be $\sum_{l=0}^L N_l/h_l$ and the variance is given by
%\begin{equation}
%V \lesssim \sum_{l=0}^L \frac{h_l^\beta}{N_l}
%\label{eq:er1smc}
%\end{equation}
Using the optimal allocation $N_l \propto \sqrt{V_l/C_l}$, one finds that  
$N_l \propto h_l^{(\beta+\zeta)/2}$.
Taking under consideration a target error of size $\mathcal{O}(\epsilon)$,
one sets $N_l \propto \epsilon^{-2} h_l^{(\beta+\zeta)/2} K_L$, 
with $K_L$  chosen to control the total error for increasing~$L$. 
Thus, for the resulted estimator in (\ref{eq:multi})-(\ref{eq:mse}), we have:
\begin{align*}
\textrm{Variance} &= \sum_{l=0}^{L} V_l N_l^{-1}=\epsilon^2 K_L^{-1} \sum_{l=0}^L h_l^{(\beta-\zeta)/2}\ ;  \\
\textrm{Comp. Cost} & =  \sum_{l=0}^L N_l C_l =  K_L^2 \epsilon^{-2}\ . 
\end{align*}
%$$
%V \lesssim \epsilon^2 K^{-1} \sum_{l=0}^L h_l^{(\beta-\zeta)/2}
%$$
To have a variance of $\mathcal{O}(\epsilon^2)$, one sets 
$K_L =  \sum_{l=0}^L h_l^{(\beta-\zeta)/2}$, so 
%The total cost of this computation is then 
%$$C=\sum_{l=0}^L N_l h_l^{-\zeta} =  K^2 \epsilon^{-2},$$
$K_L$ may or may not depend on $\epsilon$ depending on whether
this sum  converges or not (recalling that $L=\mathcal{O}(|\log(\epsilon)|)$).
In the case of Euler-Marayuma, for example, $\beta=\zeta$, $K_L=L$, 
and the cost is $\cO(\log(\epsilon)^2 \epsilon^{-2})$, versus $\cO(\epsilon^{-3})$
using a single level with mesh-size $h_L=\mathcal{O}(\epsilon)$.  If $\beta>\zeta$, 
corresponding for instance to the Milstein method, then the cost is $\cO(\epsilon^{-2})$. 
The latter is the cost of obtaining the given level of error for a scalar random variable, and
is therefore optimal.  The worst scenario is when $\beta<\zeta$.  In this case 
it is sufficient to set $K_{L}=h_L^{(\beta-\zeta)/2}$ to make the variance 
$\mathcal{O}(\epsilon^2)$, and then 
the number of samples on the finest level is given by $N_L = h_L^{\beta-2\alpha}$
whereas the total algorithmic cost is
$\cO(\epsilon^{-(\zeta/\alpha + \delta)})$, where $\delta = 2-\beta/\alpha \geq 0$.
In this case, one can choose the largest value for the bias, $\alpha = \beta/2$, so that $N_L=1$ and the total cost, $\cO(\epsilon^{-\zeta/\alpha})$,  is dominated by this single
sample.  See \cite{gile:08} for more details.

It is important to note that the realizations $U_l^{(i)}$, $U_{l-1}^{(i)}$ for a given 
increment 
must be coupled to obtain decaying variances $V_l$.  
In the case of an SDE driven by Brownian motion one can simply simulate the 
driving noise on level $l$ and then upscale it to level $l-1$ by summing elements of the finer path 
\cite{gile:08}.
For the case of a PDE forward model relying on uncertain input the scenario is quite similar \cite{cliffe2011multilevel}.  
For example, in the case 
that
the input is of fixed dimension and 
the levels arise due to discretization 
of the forward map alone within a finite element context, 
one would use the same realization of the input on two separate meshes for a pairwise-coupled realization.   
%In case that one is approximating an infinite-dimensional input which is also defined on a hierarchy of meshes, then it may be possible to obtain the coupling by simulating on the fine mesh and upscaling, similarly to the SDE case (although this may not be so straightforward).  
Note that in the more general context of PDE, it is natural to decompose 
$\zeta=d\cdot \gamma$, where $d$ is the spatio-temporal dimension of the underlying continuum.  In particular,
the number of degrees of freedom of a $d-$dimensional field approximated on a mesh of diameter $h_l$ 
is given by $h_l^{-d}$.  Then, the forward solve associated to the evaluation of $g(U_l)$ may range from linear 
($\gamma=1$) to cubic ($\gamma=3$) in the number of degrees of freedom.  For example, the solution of an SDE 
or a sparse matrix-vector multiplication give $\gamma=1$, a dense matrix-vector multiplication would give
$\gamma=2$, and direct linear solve by Gaussian elimination would give $\gamma=3$.  
%XXX [Kody: is this par. OK?]
 
%***ALL THE ABOVE NEED CORRECTIONS/ADJUSTMENTS.
%Mention that  SDEs is an easier forward problem, but for the inverse problem we are looking at, things are tougher in terms of coupling.****

The present work will focus on the case of an inverse problem with fixed-dimensional input.  Indeed the difficulty arises here because we only know how to {\it evaluate} (up-to a constant) the target density at any given level, and cannot directly obtain independent samples from it.  There exist many approaches to solving such problem, for example one can review the recent works \cite{hoan:12, ketelsen2013hierarchical} which use Markov chain Monte Carlo (MCMC) methods in the multilevel framework.  In this article a more natural 
and powerful
formulation is considered, related with the use of Sequential Monte Carlo approaches. 

Sequential Monte Carlo (SMC) methods are amongst the most widely used computational techniques in statistics, engineering, physics, finance and many other disciplines. In particular SMC samplers \cite{delm:06b} are designed to approximate a sequence $\{ \eta_l \}_{l \geq 0}$ of probability distributions on a common space, whose densities are only known up-to a normalising constant. 
The method uses $N\geq 1$
samples (or particles) that are generated in parallel, and are propagated with importance sampling (often) via MCMC and resampling methods. Several convergence results, as $N$ grows, have been proved (see  e.g.~\cite{cl-2013,delm:04,delmoral1,douc}).
SMC samplers have also recently been proven to be stable in certain high-dimensional contexts \cite{beskos}. Current state of the art for the analysis of SMC algorithms include the work
of \cite{cl-2013,chopin1,delm:04,delmoral1,douc}. In this work, the method of SMC samplers is perfectly designed to approximate the sequence of distributions, but as we will see, implementing the standard telescoping
identity of MLMC requires some ingenuity. In addition, in order to consider the benefit of using SMC, one must analyze the variance of the estimate; in such scenarios this is not a trivial extension of the convergence analysis previously mentioned.
In particular, one must very precisely consider the auto-covariance of the SMC approximations and consider the rate of decrease of this quantity as the time-lag between SMC approximations increases. Such a precise analysis does not appear to
exist in the literature. We note that our work, whilst presented in the context of PDEs, is not restricted to such scenarios and, indeed can be applied in almost any other similar context (that is,  a sequence of distributions on a common space, with increasing computational costs associated to the evaluation of the densities which in some sense converge to a given density); however, the potential benefit of doing so, may not be obvious in general.

This article is structured as follows. In Section \ref{sec:set_up} the ML identity and SMC algorithm are given. In Section \ref{sec:complex} our main complexity result is given under assumptions and their implications are discussed. In Section \ref{sec:IP}
we give a context where the assumptions of our theoretical results can be verified. In Section \ref{sec:numerics} our approach is numerically demonstrated on a Bayesian inverse problem. Section \ref{sec:complex} and the Appendix provide the proofs of our main theorem.

\section{Sequential Monte Carlo Methods}\label{sec:set_up}

\subsection{Notations}

Let $(E,\mathcal{E})$ be a measurable space.
The notation $\testfunction(E)$ denotes the class of bounded and measurable real-valued functions. The 
supremum norm is written as $\|f\|_{\infty} = \sup_{u\in E}|f(u)|$ %\sup_{x\in E}|f(x)|$ 
and $\mathcal{P}(E)$ is the set of probability measures on $(E,\mathcal{E})$. We will consider non-negative operators 
%$K : E \times \mathcal{E} \rightarrow \bbR_+$ such that for each $x \in E$ the mapping $A \mapsto K(x, A)$ is a finite non-negative measure on $\mathcal{E}$ and for each $A \in \mathcal{E}$ the function $x \mapsto K(x, A)$ is measurable; the kernel $K$ is Markovian if $K(x, dy)$ is a probability measure for every $x \in E$.
%For a finite measure $\mu$ on $(E,\mathcal{E})$, real-valued and measurable $f:E\rightarrow\mathbb{R}$
%%
%\begin{equation*}
%    \mu K  : A \mapsto \int K(x, A) \, \mu(dx)\ ;\quad 
%    K f :  x \mapsto \int f(y) \, K(x, dy).
%\end{equation*}
%We also write $\mu(f) = \int f(x) \mu(dx)$.
$K : E \times \mathcal{E} \rightarrow \bbR_+$ such that for each $u \in E$ the mapping $A \mapsto K(u, A)$ is a finite non-negative measure on $\mathcal{E}$ and for each $A \in \mathcal{E}$ the function $u \mapsto K(u, A)$ is measurable; the kernel $K$ is Markovian if $K(u, dv)$ is a probability measure for every $u \in E$.
For a finite measure $\mu$ on $(E,\mathcal{E})$,  and a real-valued, measurable $f:E\rightarrow\mathbb{R}$, we define the operations:
\begin{equation*}
    \mu K  : A \mapsto \int K(u, A) \, \mu(du)\ ;\quad 
    K f :  u \mapsto \int f(v) \, K(u, dv).
\end{equation*}
We also write $\mu(f) = \int f(u) \mu(du)$. In addition $\|\cdot\|_{r}$, $r\geq 1$, denotes the $L_r-$norm, where the expectation is w.r.t.~the law of the appropriate simulated algorithm.

\subsection{Algorithm}

As described in Section \ref{sec:intro}, the context of interest is when a sequence of densities 
$\{\eta_{l}\}_{l\ge 0}$, as in (\ref{eq:target}), are
associated to an `accuracy' parameter $h_l$, with $h_l\rightarrow 0$ 
as $l\rightarrow \infty$, such that $\infty>h_0>h_1\cdots>h_{\infty}=0$.  In practice one cannot treat $h_\infty=0$ and so must consider these distributions with $h_l>0$.
%For increasing $l$, the $\eta_l$ become more 
%accurate but more expensive to sample. 
The laws with large $h_l$ are easy to sample from with low computational cost, but are very different from $\eta_{\infty}$, whereas, those distributions with small $h_l$ are
hard to sample with relatively high computational cost, but are closer to $\eta_{\infty}$. 
% This is similar to the context of the multilevel Monte Carlo (MLMC)
%method \cite{gile:08}.
Thus, we choose a maximum level $L\ge 1$ and we will estimate
%Let $g:E\rightarrow\mathbb{R}$ be a measurable function such that $\|g\|_{\infty}=\sup_{u\in E}|g(u)|=1$ (the latter is not required, but simplifies the derivations) and suppose we are interested in
$$
\mathbb{E}_{\eta_L}[g(U)] := \int_E g(u)\eta_L(u)du\ .
$$
By the standard telescoping identity used in MLMC, one has
\begin{align}
\mathbb{E}_{\eta_L}[g(U)] & =  \mathbb{E}_{\eta_0}[g(U)] + \sum_{l=1}^{L}\Big\{
\mathbb{E}_{\eta_l}[g(U)] - \mathbb{E}_{\eta_{l-1}}[g(U)]\Big\} \nonumber \nonumber \\ 
& =\mathbb{E}_{\eta_0}[g(U)] + \sum_{l=1}^{L}\mathbb{E}_{\eta_{l-1}}\Big[
\Big(\frac{\gamma_l(U)Z_{l-1}}{\gamma_{l-1}(U)Z_l} - 1\Big)g(U)\Big]\ .
\label{eq:ml_approx}
\end{align}

Suppose now that one applies an SMC sampler \cite{delm:06b} to obtain 
a collection of samples (particles) that sequentially approximate $\eta_0, \eta_1,\ldots, \eta_L$. 
We consider the case when one initializes the population of particles by sampling  i.i.d.~from $\eta_0$, then at every step  resamples and applies a MCMC kernel to mutate the particles.
We denote by $(U_{0}^{1:N_0},\dots,U_{L-1}^{1:N_{L-1}})$, with $+\infty > N_0\geq N_1\geq \cdots  N_{L-1}\geq 1$, the samples after mutation; one resamples $U_l^{1:N_l}$ according to the weights $G_{l}(U_l^i) = 
(\gamma_{l+1}/\gamma_l)(U_l^{i})$, for indices $l\in\{0,\dots,L-1\}$.
We will denote by $\{M_l\}_{1\leq l\leq L-1}$ the sequence of MCMC kernels used at stages $1,\dots,L-1$, such that $\eta_{l}M_l = \eta_l$.
For $\varphi:E\rightarrow\mathbb{R}$, $l\in\{1,\dots,L\}$, we have the following estimator 
of $\bbE_{\eta_{l-1}}[\varphi(U)]$:
$$
\eta_{l-1}^{N_{l-1}}(\varphi) = \frac{1}{N_{l-1}}\sum_{i=1}^{N_{l-1}}\varphi(U_{l-1}^i)\ . 
$$
%and $G_l(u) = \gamma_{l+1}(u)/\gamma_{l}(u)$, $l\in\{0,\dots,L-1\}$.
We define
$$
\eta_{l-1}^{N_{l-1}}(G_{l-1}M_l(du_l)) = \frac{1}{N_{l-1}}\sum_{i=1}^{N_{l-1}}G_{l-1}(U_{l-1}^i) M_l(U_{l-1}^i,du_l)\ .
$$
The joint probability distribution for the SMC algorithm is 
$$
\prod_{i=1}^{N_0} \eta_0(du_0^i) \prod_{l=1}^{L-1} \prod_{i=1}^{N_l} \frac{\eta_{l-1}^{N_{l-1}}(G_{l-1}M_l(du_l^i))}{\eta_{l-1}^{N_{l-1}}(G_{l-1})}\ .
$$
If one considers one more step in the above procedure, that would deliver samples 
$\{U_L^i\}_{i=1}^{N_L}$, a standard SMC sampler estimate of the quantity of interest in (\ref{eq:ml_approx})
is $\eta_L^{N}(g)$; the earlier samples are discarded. 
Within a multilevel context, a consistent SMC estimate of \eqref{eq:ml_approx}
is
\begin{equation}
\widehat{Y} =
\eta_{0}^{N_0}(g) + \sum_{l=1}^{L}\Big\{\frac{\eta_{l-1}^{N_{l-1}}(gG_{l-1})}{\eta_{l-1}^{N_{l-1}}(G_{l-1})} - \eta_{l-1}^{N_{l-1}}(g)\Big\}\label{eq:smc_est}\ , 
\end{equation}
and this will be proven to be superior than the standard one, under assumptions.

There are two important structural differences within the MLSMC context, 
compared to the standard ML implementation of \cite{gile:08}, sketched in 
Section \ref{sec:intro}:
\begin{itemize}
\item[i)] the   $L+1$ terms  in (\ref{eq:smc_est}) are \emph{not} unbiased estimates 
of the differences $\mathbb{E}_{\eta_l}[g(U)] - \mathbb{E}_{\eta_{l-1}}[g(U)]$,
so the relevant MSE error decomposition here is:
\begin{equation}
\label{eq:dec}
\mathbb{E}\big[ \{\widehat{Y}-\mathbb{E}_{\eta_\infty}[g(U)] \}^2\big]
\le 2\,\mathbb{E}\big[\{\widehat{Y}-\mathbb{E}_{\eta_L}[g(U)]\}^2\big] + 
2\,\{ \mathbb{E}_{\eta_L}[g(U)] - \mathbb{E}_{\eta_\infty}[g(U)]  \}^2\ . 
\end{equation} 
\item[ii)] the same $L+1$ estimates are \emph{not} independent. Hence a
substantially more complex error analysis will be required to characterise  
$\mathbb{E}[\{\widehat{Y}-\mathbb{E}_{\eta_L}[g(U)]\}^2]$.
In Section \ref{sec:complex}, we will obtain an expression for this discrepancy, 
which will be more involved than the standard $\sum_{l=0}^{L}V_l/N_l$, 
but will still allow for a relevant constrained optimisation 
to determine the optimal allocation of particle sizes $N_l$ along the levels.
\end{itemize}
Given an appropriate classification of both terms on the R.H.S.\@ 
of (\ref{eq:dec}) as an order of the tolerance 
for a Bayesian Inverse Problem (to be described in Section \ref{sec:IP}), one can specify a level $L$, and optimal
Monte-Carlo sample sizes $N_l$ so that the MSE of $\widehat{Y}$ 
is $\mathcal{O}(\epsilon^2)$ at a reduced computational cost.

\section{Development of multilevel SMC}\label{sec:complex}

\subsection{Main Result}

We will now obtain an analytical result that controls the error term 
$\mathbb{E}[\{\widehat{Y}-\mathbb{E}_{\eta_L}[g(U)]\}^2]$ in expression (\ref{eq:dec}).
This is of general significance for the development of MLSMC in various contexts.
Then, we will look in detail at an inverse problem context (developed in Section \ref{sec:IP})
and fully investigate the MLSMC method.

For any $l\in\{0,\dots,L\}$ and $\varphi\in \mathcal{B}_b(E)$ we write:
$
\eta_l(\varphi) := \int_E \varphi(u)\eta_l(u)du.
$
We introduce the following assumptions, which will be verifiable in some contexts. They are rather strong, but could be relaxed at condsiderable increase in the complexity of the arguments,
which will ultimately provide the same information. In addition, the assumptions are standard in the literature of SMC methods; see \cite{delm:04,delmoral1}.

\begin{hypA}
\label{hyp:A}
There exist $0<\underline{C}<\overline{C}<+\infty$ such that
\begin{eqnarray*}
\sup_{l \geq 1} %1\leq l< \infty} 
\sup_{u\in E} G_l (u) & \leq & \overline{C}\ ;\\
\inf_{l \geq 1} %1\leq l< \infty} 
\inf_{u\in E} G_l (u) & \geq & \underline{C}\ .
\end{eqnarray*}
\end{hypA}

\begin{hypA}
\label{hyp:B}
There exists a $\rho\in(0,1)$ such that for any $l\ge 1$, $(u,v)\in E^2$, $A\in\mathcal{E}$:
$$
\int_A M_l(u,du') \geq \rho \int_A M_l(v,dv')\ .
$$
\end{hypA}

\begin{theorem}\label{theo:main_error}
Assume (A\ref{hyp:A}-\ref{hyp:B}). There exist $C<+\infty$ and  $\kappa\in (0,1)$ such  that for any $g\in\mathcal{B}_b(E)$, with $\|g\|_{\infty}=1$,
\begin{align*}
\mathbb{E}\big[\{\widehat{Y}-\mathbb{E}_{\eta_L}[g(U)]\}^2\big] 
\leq 
C\,\bigg(\frac{1}{N_0} + &\sum_{l=1}^{L}\frac{\|\tfrac{Z_{l-1}}{Z_{l}}G_{l-1}-1\|_{\infty}^2}{N_{l-1}} \\ &+ 
\sum_{1\le l<q\le L}\|\tfrac{Z_{l-1}}{Z_{l}}G_{l-1}-1\|_{\infty} \|\tfrac{Z_{q-1}}{Z_{q}}G_{q-1}-1\|_{\infty}
\big\{\tfrac{\kappa^{q-l}}{N_{l-1}}
+\tfrac{1}{N_{l-1}^{1/2}N_{q-1}} 
\big\}\bigg)\ .
\end{align*}

\end{theorem}

\subsection{Proof of Theorem \ref{theo:main_error}}\label{sec:main_proof}

The following notations are adopted; this will substantially simplify subsequent expressions:
\begin{align}
Y_{l-1}^{N_{l-1}} &= \frac{\eta_{l-1}^{N_{l-1}}(gG_{l-1})}{\eta_{l-1}^{N_{l-1}}(G_{l-1})} - \eta_{l-1}^{N_{l-1}}(g)\ , \quad  \nonumber \\[0.2cm] 
Y_{l-1} &= \frac{\eta_{l-1}(gG_{l-1})}{\eta_{l-1}(G_{l-1})} - \eta_{l-1}(g)
\,\,\,\,\big(\,  \equiv \eta_{l}(g) - \eta_{l-1}(g)\, \big)\ , \label{eq:analytical} \\[0.3cm]
\nonumber
\overline{\varphi_l}(u) & = \big(\tfrac{Z_{l-1}}{Z_l}G_{l-1}(u)-1\big) \ , \\[0.3cm]
\nonumber
\widetilde{\varphi_l}(u)
&=  g(u) \overline{\varphi_l}(u) 
\ , \\[0.3cm]
\label{eq:ay}
A_n(\varphi,N) &  = \eta_n^N(\varphi G_n)/\eta_n^N(G_n) \ , \quad \varphi\in\mathcal{B}_b(E)\ , \quad 
0\leq n\leq L-1  \ , \\[0.2cm]
\label{eq:aybar}
\overline{A}_n(\varphi,N) &  =  A_n(\varphi,N) - \frac{\eta_n(\varphi G_n)}{\eta_n(G_n)}\ .
\end{align}
Throughout this Section, $C$ is a constant whose value may change, but does not depend on any time parameters of the Feynman-Kac formula, nor $N_l$. The proof of Theorem \ref{theo:main_error} follows from several technical lemmas which are now given and supported by further results in the Appendix; the proof of the theorem is at the end of this subsection.
It is useful to observe that $Z_l/Z_{l-1} = \eta_{l-1}(G_{l-1})$, $\eta_{l-1}(\overline\varphi_l) =0$
and 
$|A_n(\varphi,N)|\le |\varphi|_{\infty}$ with probability 1 
as the conditional $L_1$-norm of functional $\varphi$
over a discrete distribution.
We will make repeated use of the following identity which follows from these observations upon adding and subtracting
$ \eta_{l-1}^{N_{l-1}} (\frac{Z_{l-1}}{Z_l}g(\cdot)G_{l-1}(\cdot) )$:
\begin{equation}
\label{eq:basic}
Y_{l-1}^{N_{l-1}}-Y_{l-1} = A_{l-1}(g,N_{l-1})\,\{\eta_{l-1} - \eta_{l-1}^{N_{l-1}}\} (\overline{\varphi_l})  + \{\eta_{l-1}^{N_{l-1}}-\eta_{l-1}\} (\widetilde{\varphi_l})\ .
\end{equation}
\begin{lem}\label{lem:tech_lem}
Assume (A\ref{hyp:A}-\ref{hyp:B}). There exists a $C<+\infty$ such that 
for any $l\ge 1$:
$$
\|Y_{l-1}^{N_{l-1}}- Y_{l-1} \|_2^2 \leq \frac{C\,\|\frac{Z_{l-1}}{Z_{l}}G_{l-1}-1\|_{\infty}^2}{N_{l-1}}\ .
$$
\end{lem}

\begin{proof}
From (\ref{eq:basic}) and the $C_2$-inequality we obtain:
\begin{align*}
\|Y_{l-1}^{N_{l-1}}- Y_{l-1} \|_2^2 \le   
2\,\|A_{l-1}(g,N_{l-1})\{\eta_{l-1}^{N_{l-1}}-\eta_{l-1}\} (\overline{\varphi_l})\|^2_2 + 
2\,\|\{\eta_{l-1}^{N_{l-1}}-\eta_{l-1}\} (\widetilde{\varphi_l})\|^2_2
\\ 
\leq 2\,\|\{\eta_{l-1}^{N_{l-1}}-\eta_{l-1}\} (\overline{\varphi_l})\|_2^2
+ 2\,\|\{\eta_{l-1}^{N_{l-1}}-\eta_{l-1}\} (\widetilde{\varphi_l})\|^2_2
\end{align*}
By \cite[Theorem 7.4.4]{delm:04} we have that both $L_2$-norms are upper bounded by
$\frac{C\|\frac{Z_{l-1}}{Z_{l}}G_{l-1}-1\|_{\infty}^2}{2N_{l-1}}$.
This  completes the proof.
\end{proof}

By the $C_2$-inequality and standard properties of i.i.d.~random variables one has: 
\begin{align*}
\mathbb{E}\big[\{\widehat{Y}-\mathbb{E}_{\eta_L}[g(U)]\}^2\big]
 = \mathbb{E}\Big[\big\{\sum_{l=1}^{N}(Y_{l-1}^{N_{l-1}} - Y_{l-1})\big\}^2\Big]
 \le \frac{C}{N_0} + 2\,\mathbb{E}\Big[\big\{\sum_{l=2}^{N}(Y_{l-1}^{N_{l-1}} - Y_{l-1})\big\}^2\Big]\ .
\end{align*}
We have that: 
\begin{equation*}
\mathbb{E}\Big[\big\{\sum_{l=2}^{N}(Y_{l-1}^{N_{l-1}} - Y_{l-1})\big\}^2\Big]
= \mathbb{E}\Big[ \sum_{l=2}^{N}(Y_{l-1}^{N_{l-1}} - Y_{l-1})^2\Big] + 2\sum_{2\le l < q\le L}  \mathbb{E}\big[(Y_{l-1}^{N_{l-1}} - Y_{l-1}) (Y_{q-1}^{N_{q-1}} - Y_{q-1})\big]
\end{equation*}
Lemma \ref{lem:tech_lem} gives that: 
\begin{equation*}
 \mathbb{E}\Big[ \sum_{l=2}^{N}(Y_{l-1}^{N_{l-1}} - Y_{l-1})^2\Big]\le 
C\sum_{l=2}^{L}\frac{\|\frac{Z_{l-1}}{Z_{l}}G_{l-1}-1\|_{\infty}^2}{N_{l-1}}
\end{equation*}
thus it remains to treat the cross-interaction terms.
Using the decomposition in (\ref{eq:basic}), we obtain
\begin{align*}
\sum_{2\le l < q\le L} \mathbb{E}&\big[(Y_{l-1}^{N_{l-1}} - Y_{l-1}) (Y_{q-1}^{N_{q-1}} - Y_{q-1})\big] =   \\
&=\sum_{2\le l < q\le L}  \mathbb{E}\,\big[A_{l-1}(g,N)A_{q-1}(g,N)\{\eta_{l-1}^{N_{l-1}}-\eta_{l-1}\}(\overline{\varphi_l})\{\eta_{q-1}^{N_{q-1}}-\eta_{q-1}\}(\overline{\varphi_q})\,\big]\\
&\hspace{1.5cm}+\sum_{2\le l < q\le L}
\mathbb{E}\,\big[\,A_{l-1}(g,N)\{\eta_{l-1}^{N_{l-1}}-\eta_{l-1}\}(\overline{\varphi_l})\{\eta_{q-1}^{N_{q-1}}-\eta_{q-1}\}(\widetilde{\varphi_q})\,\big]\\
&\hspace{1.5cm}+\sum_{2\le l < q\le L}
\mathbb{E}\,\big[\,A_{q-1}(g,N)\{\eta_{l-1}^{N_{l-1}}-\eta_{l-1}\}(\widetilde{\varphi_l})\{\eta_{q-1}^{N_{q-1}}-\eta_{q-1}\}(\overline{\varphi_q})\,\big]\\
&\hspace{1.5cm}+\sum_{2\le l < q\le L}
\mathbb{E}\,\big[\,\{\eta_{l-1}^{N_{l-1}}-\eta_{l-1}\}(\widetilde{\varphi_l})\{\eta_{q-1}^{N_{q-1}}-\eta_{q-1}\}(\widetilde{\varphi_q})\,\big]\ .
\end{align*}
We will now apply Proposition \ref{prop:prop_corr_bd3} to the relevant terms in the sum, to yield the upper-bound:
\begin{align*}
C \sum_{1\le l < q\le L}\|\widetilde{\varphi_l}\|_{\infty}\|\widetilde{\varphi_q}\|_{\infty}\Big\{\frac{\kappa^{q-l}}{N_{l-1}}
&+\frac{1}{N_{l-1}^{1/2}N_{q-1}}
\Big\}\ .
\end{align*}
From here one can conclude the proof of Theorem \ref{theo:main_error}.

\subsection{MLSMC Variance Analysis}
\label{ssec:multilevel_component}

This section considers the specification of parameters for the MLSMC algorithm after consideration of Theorem \ref{theo:main_error}.  Recall that in the simpler SDE setting of \cite{gile:08} 
one must work with the strong error estimate $\bbE |g(U_l) - g(U)|^2 = \cO(h_l^\beta)$
and the deduced variance $V_l=\mathrm{Var}[g(U_l)-g(U_{l-1})] = \cO(h_l^{\beta})$.
From Theorem \ref{theo:main_error}, a similar role within MLSMC is taken 
by: 
\begin{equation}
\label{eq:VL}
V_l:= \|\tfrac{Z_{l-1}}{Z_{l}}G_{l-1}-1\|_{\infty}^2 \ .
\end{equation}
We assume that in the given context one can obtain 
that $V_l = \mathcal{O}(h_l^{\beta})$ for some appropriate
rate constant $\beta\ge 1$. 
Recall that we have $h_l=M^{-l}$, for some integer $M>1$
and we assume a bias of $\mathcal{O}(h_L^{\alpha})$.
Thus, targeting an error tolerance of $\epsilon$, we have
 $h_L^{\alpha} = M^{-L}=\mathcal{O}(\epsilon)$, 
so that  $L=\log(\epsilon^{-1})/(\alpha \log(M)) + \cO(1)$.
Now, to   optimally allocate $N_0, N_1, \ldots, N_L$,
one  proceeds along the lines outlined in the Introduction
under consideration of Theorem \ref{theo:main_error}. 
Notice that $\sum_{q=l+1}^L \kappa^{q-l} \leq \frac{1}{1-\kappa}$ and 
$V_q$ is smaller  than $V_l$ (in terms of the obtained upper bounds), so the upper bound in Theorem \ref{theo:main_error} can be bounded by:
\begin{equation}
\label{eq:up}
\frac{1}{N_0} + \sum_{l=1}^L \bigg(\frac{h_l^\beta}{N_l} + 
\Big(\frac{h_l^\beta}{N_l} \Big)^{1/2} \sum_{q=l+1}^L \frac{h_q^{\beta/2}}{N_q} \bigg)\ .
\end{equation}
We also assume a computational cost proportional to $\sum_{l=0}^L N_l h_l^{-\zeta}$, for some rate $\zeta\ge 1$, 
with the resampling cost considered to 
to be negligible for practical purposes compared to the cost of the  calculating the importance weights (as it is the case for the inverse problems we focus upon later).
As with standard MLMC in \cite{gile:08}, we need to find $N_{0},\ldots, N_L$ 
that optimize (\ref{eq:up}) given a fixed computational cost $\sum_{l=0}^L N_l h_l^{-\zeta}$.
Such a constrained optimization with the complicated error bound in (\ref{eq:up}), results
in the need to solve
a quartic equation
as a function of $V_l$ and $C_l$.  Instead, one can {\it assume} that the second term
on the R.H.S.\@ of (\ref{eq:up}) 
 is negligible, solve the constrained optimization ignoring that term, 
and then check that the effect of that term for the given choice of $\{N_l\}_{l=0}^{L-1}$ is smaller than $\mathcal{O}(\epsilon^2)$. 
Following this approach gives a constrained optimisation problem 
identical to the simple case of \cite{gile:08}, with solution
$N_l \propto \sqrt{V_l/C_l} = \mathcal{O}(h_l^{(\beta+\zeta)/2})$.
%Thus, 
One works as in Section \ref{sec:intro}, and 
selects:
\begin{equation*}
N_l \propto \epsilon^{-2}h_l^{(\beta+\zeta)/2}K_L \ ; \quad 
K_L  \eqsim \sum_{l=0}^{L} h_l^{(\beta-\zeta)/2}\ . 
\end{equation*}
Then returning to (\ref{eq:up}) one can check that indeed the extra summand 
is smaller than $\mathcal{O}(\epsilon^2)$ for the above choice 
of $N_l$. Notice that: (i)\, $h_q^{\beta/2}/N_q =
\mathcal{O}(\epsilon^{2}h_l^{-\zeta/2}/K_L)$, 
and the sum $\sum_{q=l+1}^{L}h_l^{-\zeta/2}$ is dominated 
by $h_L^{-\zeta/2} = \mathcal{O}(\epsilon^{-\zeta/(2\alpha)})$;
\,(ii)\,we have $(h_l^{\beta}/N_l)^{1/2}   \propto \epsilon/K_L^{1/2}h_l^{(\beta-\zeta)/4} $. 
Therefore, %
\begin{align*}
\sum_{l=1}^L \bigg(
\Big(\frac{h_l^\beta}{N_l} \Big)^{1/2} \sum_{q=l+1}^L \frac{h_q^{\beta/2}}{N_q} \bigg) = \mathcal{O}\Big(\epsilon^{2}
\epsilon^{1-\zeta/(2\alpha)}\sum_{l=0}^{L}h_l^{(\beta-\zeta)/4}
/K_L^{3/2}\Big) = \mathcal{O}(\epsilon^{2}\epsilon^{1-\zeta/(2\alpha)})\ . 
\end{align*}
Thus, when $\zeta\le 2\alpha$, the overall mean squared error 
is still $\mathcal{O}(\epsilon^2)$.
In the inverse problem context of Section \ref{sec:IP},
we will establish that $\beta=2$, $\alpha=\beta/2$.
Also, in many cases (depending on the chosen PDE solver) 
we have $\zeta=d$.

\section{Bayesian Inverse Problem}
\label{sec:IP}

A context will now be introduced in which the results are of interest and 
the assumptions can be satisfied.  We begin with another round of notations.
  Introduce the Gelfand triple $V := H^{1}(D) \subset L^2(D) \subset H^{-1} (D)=: V^*$, 
where the domain $D$ will be understood.  
Furthermore, denote by $\langle \cdot, \cdot \rangle, \|\cdot\|$ the inner product and norm 
on $L^2$, with superscripts 
to denote the corresponding inner product and norm on the Hilbert 
spaces $V$ and $V^*$.  Denote the finite dimensional Euclidean inner product and norms as 
$\langle \cdot, \cdot \rangle, |\cdot|$, with the latter also representing size of a set and absolute value, 
and denote weighted norms by adding a subscript   as  
$\langle,\cdot, \cdot \rangle_A := \langle A^{-\frac12}\cdot, A^{-\frac12}\cdot \rangle$, with corresponding norms
$|\cdot |_A$ or $\|\cdot \|_A$ for Euclidean and $L^2$ spaces, respectively 
(for symmetric, positive definite $A$ with $A^\frac12$ being the unique symmetric square root).
In the following, the generic constant $C$ will be used for the right-hand side of inequalities as necessary, 
its precise value actually changing between usage.

Let $D \subset \bbR^d$ with $\partial D \in C^1$ convex.
For $f \in V^*$, consider the following PDE on $D$:
\begin{align}
\label{eq:uniellip}
- \nabla \cdot ( \hu \nabla p )  & =  f\ , \quad  {\rm on} ~ D\ , \\
p & = 0\ , \quad  {\rm on} ~ \partial D\ ,
\label{eq:bv}
\end{align}
where:
\begin{equation}
\label{eq:expand}
\hu (x) = \bar{u}(x) + \sum_{k=1}^K u_k \sigma_k \phi_k(x) \ . 
\end{equation}
Define $u=\{u_k\}_{k=1}^K$, with $u_k \sim 
U[-1,1]$ i.i.d. This determines the prior distribution for $u$.
  Assume that $\bar{u}, \phi_k \in C^\infty$ for all $k$ and that  
$\|\phi_k\|_\infty =1$.  
 In particular, 
 assume $\{\sigma_k\}_{k=1}^K$ 
 decay\footnote{If $K\rightarrow \infty$ it is important that they decay with a suitable 
 rate in order to ensure $u$ lives almost surely in an appropriate sequence-space,
or equivalently $\hu$ lives in the appropriate function-space.  
However, here we down-weight higher frequencies as necessary only to 
induce certain smoothness properties, while actually for a given value of $u \in E$ 
the resulting permeability 
$\hu \in \widehat{E} \subset C^\infty(D) \subset C(D) \subset L^\infty(D) \subset L^p(D)$ for all $p\geq 1$.} with $k$.  
The state space is $E = \prod_{k=1}^K [-1,1]$.  
It is important that the following property holds:
$$\inf_x \hu(x) \geq  \inf_x \bar{u}(x) - \sum_{k=1}^K \sigma_k \geq u_* > 0$$
so that the operator on the 
left-hand side of \eqref{eq:uniellip} is uniformly elliptic.  
Let $p(\cdot;u)$ denote the weak solution of \eqref{eq:uniellip} for parameter value $u$.  
Define the following the vector-valued function 
$$
\cG(p) = [ g_1( p), \cdots , g_M ( p ) ]^\top\ , 
$$
where $g_m$ are elements of the dual space 
$V^*$ for  $m=1,\ldots, M$.
It is  assumed that the data take the form
\begin{equation}
y = \cG (p) + \xi\ , \quad \xi \sim N(0,\Gamma)\ , \quad \xi \perp u\ , 
\label{eq:data}
\end{equation}
where $N(0,\Gamma)$ denotes the Gaussian random variable with mean $0$ and covariance $\Gamma$, 
and $\perp$ denotes independence. 
The unnormalized density then is given by:
\begin{equation*}
\gamma(u) = e^{-\Phi[\cG(p(\cdot;u))]}  \ ; \quad \Phi(\cG) = \tfrac{1}{2}\, | \cG - y|^2_\Gamma
\ . 
\end{equation*}

Consider the triangulated domains $\{D^l\}_{l=1}^\infty$ 
approximating $D$, 
where $l$ indexes the number of nodes $N(l)$, so that  we have 
$D^1 \subset\cdots \subset  D^{l} \subset D^\infty :=D$,
with sufficiently regular triangles.
Consider a finite element discretization on $D^l$ 
consisting of $H^{1}$ functions $\{\psi_\ell\}_{\ell=1}^{N(l)}$.
In particular, continuous piecewise linear hat functions will be
considered here, the explicit form of which will be given in section \ref{ssec:numset}.
Denote the corresponding space
of functions of the form $\varphi = \sum_{\ell=1}^{N(l)} v_\ell \psi^l_\ell$ by $V^l$, and notice that 
$V^1\subset V^{2}\subset \cdots \subset V^l \subset V$.  
By making the further Assumption 7 of 
\cite{hoan:12} that the weak solution $p(\cdot;u)$ of \eqref{eq:uniellip}-(\ref{eq:bv}) for parameter value $u$ 
is in the space $W=H^2 \cap H^1_0 \subset V$, one obtains a well-defined  
finite element approximation $p^l(\cdot;u)$ of $p(\cdot;u)$.
Thus, the sequence of distributions of interest in this context is:
\begin{equation*}
\gamma_l(u) = e^{-\Phi[\cG(p^l(\cdot;u))]}\ , \quad l=0,1,\ldots, L\ . 
\end{equation*}

 \subsection{Error Estimates}\label{sec:verify}

Notice one can 
take the inner product of \eqref{eq:uniellip} with the solution $p \in V$,
and perform integration by parts on the right-hand side, 
in order to obtain 
$\langle \hu \nabla p, \nabla p \rangle  =  \langle f , p \rangle$.
Therefore
\begin{equation}
u_* \| p \|^2_V = u_* \langle \nabla p, \nabla p \rangle  \leq 
 \langle \hu \nabla p, \nabla p \rangle = 
  \langle f , p \rangle \leq \|f\|_{V^*} \|p\|_V.
\end{equation}
So the following bound holds in $V$, uniformly over $u$:
\begin{equation}
 \| p(\cdot;u) \|_V  \leq \frac{\|f\|_{V^*}}{u_*}\ .
 \label{eq:pvbound}
\end{equation}
Notice that: 
\begin{equation}
|\cG(p)-\cG(p')| = \Big(\sum_{m=1}^M \langle g_m, p-p' \rangle^2 \Big)^{1/2} \leq \| p -p'\|_V
\sum_{m=1}^M \|g_m\|_{V^*} = C \| p-p' \|_V\  .
\label{eq:gunifu}
\end{equation}
So the following uniform bound also holds:
\begin{equation*}
|\cG(p(\cdot;u))| \leq C\,\frac{\|f\|_{V^*}}{u_*}\ . 
\end{equation*}
The uniform bound on $\cG$ provides the Lipschitz bound   
\begin{equation}
|\Phi(\cG) - \Phi(\cG')| \leq C |\cG - \cG'|, 
\label{eq:ctslip}
\end{equation}
obtained as follows:
\begin{align}
\nonumber
|\Phi(\cG) - \Phi(\cG')|  = & \frac12\left |  |\cG - y|_\Gamma^2 -  |\cG' - y|_\Gamma^2 \right |\\
\nonumber
 = &\left | |\cG|_\Gamma^2 - |\cG'|_\Gamma^2 + 2 \langle \cG' - \cG , y \rangle _{\Gamma} \right |\\
\nonumber
\leq & \left (  |\cG| + |\cG'| + 2 |y| \right ) |\Gamma^{-1}| |\cG - \cG'|\ ,
\end{align}
Setting $\cG'=0$ gives the boundedness of 
$\Phi$.

Considering  some sequence $h_l$ indicating the maximum diameter of an individual element
at level $l$, with $h_l \rightarrow 0$ (e.g. $h_l = 2^{-l}$),
the following asymptotic bound holds for continuous piecewise linear hat functions
\cite{ciarlet1978finite}\footnote{Higher order finite elements can yield stronger convergence rates, 
but will not be considered here in the interest of a more streamlined presentation.}
\begin{equation}
\|p(\cdot;u) - p^l(\cdot;u)\|_V \leq C h_l \| p(\cdot;u)\|_W\ . 
\label{eq:femconv}
\end{equation} 
Furthermore, Proposition 29 of \cite{hoan:12} provides a uniform bound based on the 
following decomposition of \eqref{eq:uniellip}:
\[
-\Delta p = \frac{1}{\hu} \left ( f + \nabla \hu \cdot \nabla p 
\right)\ .
\]
Thus, we have 
\begin{align}
\nonumber
{\rm sup}_{u} \|p(\cdot;u)\|_W & \leq  C' {\rm sup}_{u} \|\Delta p(\cdot;u)\| \\
\nonumber
& \leq  \frac{C'}{u_*} {\rm sup}_{u}\left( \|f\| + \|\hu\|_V \|p\|_V \right ) \\
& \leq  C \|f\|\ ,
\label{eq:breakdownw}
\end{align}
where the first line holds by equivalence of norms, the second holds
since $\hu \in C^\infty$, by the triangle inequality and Cauchy-Schwarz inequality,
and the last line holds by \eqref{eq:pvbound} and the fact $\|f\|_{V^*} \leq c \|f\|$ for some $c$.
The constant $C$ depends on $u_*, \|\nabla \hu\|_\infty, C',$ and $c$ .
Note that $ \|\hu\|_V \leq \| \nabla \hu\|_\infty \leq C'' < \infty$ by \eqref{eq:expand}.
Note that the bound \eqref{eq:breakdownw} in \eqref{eq:femconv} 
together with \eqref{eq:pvbound} provides a uniform bound over $l$ for 
$\cG^l$, defined by $\cG^l: u \mapsto \cG(p^l(\cdot;u))$, 
following the same argument as \eqref{eq:gunifu},
which means that the Lipschitz bound in \eqref{eq:ctslip}
holds here over different $l$ as well.  

Now, the following holds by \eqref{eq:femconv}, 
%(uniformly in $u$ by) 
\eqref{eq:breakdownw}, \eqref{eq:pvbound}, and the triangle inequality
\begin{equation}
\|p^{l}(\cdot;u) - p^{l-1}(\cdot;u)\|_V \leq C h_l\ .
%C h_l \| p(\cdot;u)\|_W \leq .
\label{eq:lincrement}
\end{equation}
Hence, from (\ref{eq:gunifu}) %we have
\begin{equation}
|\cG^l(u) - \cG^{l-1}(u) | = |\cG(p^{l}(\cdot;u)) - \cG(p^{l-1}(\cdot;u)) | \leq C h_l\ ,
\label{eq:glincrement}
\end{equation}
where $C$ is independent of the realization of $u$.
\begin{prop}
\label{pr:V}
For $G_{l-1}(u) := \exp\{\Phi(\cG^{l-1}(u)) - \Phi(\cG^{l}(u))\}$ one has the following 
estimates, uniformly in $u$: %that: %we have that:
 \begin{equation}
1 - \cO(h_l) = \underline{C}_l := e^{-C h_l} \leq G_{l-1} = \exp\{\Phi(\cG^{l-1}) - \Phi(\cG^{l})\} \leq e^{C h_l} 
=: \overline{C}_l = 1 + \cO(h_l).
\label{eq:wlestimate}
\end{equation}
\end{prop}
\begin{proof}
In combination with \eqref{eq:ctslip}, equation (\ref{eq:glincrement}) gives
the stated result.
%%Remark \ref{rem:bnds}.
%%This estimate in combination with Eq. \eqref{eq:ctslip} 
%%Furthermore, this guarantees the following 
% \begin{equation}
%1 - \cO(h_l) = \underline{C}_l := e^{-C h_l} \leq G_{l-1} = \exp\{\Phi(\cG^{l-1}) - \Phi(\cG^{l})\} \leq e^{C h_l} 
%:= \overline{C}_l = 1 + \cO(h_l).
%\label{eq:wlestimate}
%\end{equation}
%thus verifying Remark \ref{rem:bnds}.
\end{proof}

\begin{prop}[Bias]
\label{pr:bias}
Let $g\in \testfunction(E)$.
Then 
$$|\mathbb{E}_{\eta_L}[g(U)] - \mathbb{E}_{\eta_\infty}[g(U)]| \le C h_L\  .$$
\end{prop}
\begin{proof}
It follows from the same reasoning as in Proposition \ref{pr:V},
upon observing that 
$$\mathbb{E}_{\eta_L}[g(U)] - \mathbb{E}_{\eta_\infty}[g(U)]
= \mathbb{E}_{\eta_\infty}\left[g(U)\left( \frac{d \eta_L}{d\eta_\infty} - 1\right)\right]\  .$$
%It follows from the above discussion .
\end{proof}

\subsection{Verification of Assumptions}
\label{sec:v}

Assumption (A\ref{hyp:A}) is satisfied by letting 
$$\underline{C} := \inf_{l \geq 1} %1 \leq l}
 \underline{C}_l\ ;  
\quad 
\overline{C} := \sup_{l \geq 1} %1 \leq l} 
\overline{C}_l\ .$$ 
Notice that the asymptotic bounds of Proposition \ref{pr:V} imply that 
$\underline{C}_l$ is increasing with $l$ while $\overline{C}_l$
are decreasing with $l$.  Therefore, these will actually be minimum and maximum
over a sufficiently large set of low indices.

%
%We will now consider 
%A context will now be introduced in which the results are of interest and %our 
%the assumptions can be satisfied.  
%In particular, subsection will concern (A\ref{hyp:B}), while subsection will concern (A\ref{hyp:A}) 
%
%
%
%For 
Assumption (A\ref{hyp:B}) %this 
can be shown to hold in this context, %the following context, 
if a Gibbs sampler is constructed. %***May need to elaborate further -- this is largely cut/paste!***.
Let $\theta$ be the uniform measure on $[-1,1]$ and consider a probability measure $\pi$ on $E:=\prod_{i=1}^{K}[-1,1]$ with density w.r.t.~the measure $\bigotimes_{i=1}^{K} \theta(du_i)$:
$$
\pi(u) = \frac{\exp\{-\Phi(u)\}}{\int_{E}\exp\{-\Phi(u)\}\bigotimes_{i=1}^{K} \theta(du_i)}
$$
where it is assumed that $\forall u\in E$, $\Phi(u)\in[0,\Phi^*]$.   
This is the setting above, for all $l$, 
following from equations \eqref{eq:ctslip} and \eqref{eq:glincrement}. 

%*** Below needs to be checked!  I attempted to modify to finite dimensions, but it should be verified if I did it correctly. ***

Let $k\in\mathbb{N}, k<K$ be given
and consider a partition of $[1,\dots, K]$ into $k$ disjoint subsets $(a_i)_{i=1}^k$.
% of natural numbers.
%$k-$natural-number-valued sets $(a_1(i))_{i\leq 1},\dots,(a_k(i))_{1\leq i \leq K}$ such that
%$$
%\bigcup_{j=1}^k\{a_j(i):1\leq i \leq K\} = \mathbb{N} \quad\quad\quad  \forall l\neq j~\{a_j(i):1\leq i \leq K\}\cap \{a_l(i):1\leq i \leq K\} = \emptyset.
%$$
For example $k=2$ and 
%$(a_1(i))_{1\leq i \leq K},(a_2(i))_{1\leq i \leq K}$ 
$a_1$ and $a_2$ are the sets of (positive) odd and even numbers up to $K$, respectively. 

One can consider the Gibbs sampler to generate from $\pi$, with kernel:
$$
M(u,du') = \Big(\prod_{j=1}^{k} \pi(u_{a_j}'|u_{a_1:a_{j-1}}',u_{a_{j+1}:a_{k}}) \Big)\bigotimes_{i=1}^{K} \theta(du_i')
$$
with 
$$
\pi(u_{a_j}'|u_{a_1:a_{j-1}}',u_{a_{j+1}:a_{k}}) = \frac{\pi(u_{a_1:a_{j}}',u_{a_{j+1}:a_{k}})}{\int_{[-1,1]^{|\{a_j\}|}} \pi(u_{a_1:a_{j}}',u_{a_{j+1}:a_{k}}) \bigotimes_{i\in(a_j)} \theta(du_i')}.
$$
%In practice, under a finite discretization $d$ of $u$, one just samples the full conditionals of each finite sub-sequence. Note that even in infinite dimensions,
%the task is not unreasonable. F
One can, for example, perform rejection sampling on $\pi$ using the prior as a proposal (and accepting with probability $\exp\{-\Phi(u)\}$)
and we would still have a theoretical acceptance probability of
$$
\int_{E}\exp\{-\Phi(u)\}\bigotimes_{i=1}^{K} \theta(du_i) \geq \exp\{-\Phi^*\}.
$$
Sampling from the full conditionals will have a higher-acceptance probability and thus the Gibbs sampler is not an unreasonable algorithm.

%We have the following result: % (which should also hold when one discretizes):

\begin{prop}
For any $u,\tilde{u}\in E$ %we have
$$
M(\tilde{u},du') \geq \exp\{-2\Phi^* (k-1)\} M(u,du').
$$
\end{prop}

\begin{proof}
Consider
\begin{eqnarray*}
\frac{\pi(u_{a_j}'|u_{a_1:a_{j-1}}',u_{a_{j+1}:a_{k}})}{\pi(u_{a_j}'|u_{a_1:a_{j-1}}',\tilde{u}_{a_{j+1}:a_{k}})}  & = & \frac{\pi(u_{a_1:a_j}',u_{a_{j+1}:a_k})}{\pi(u_{a_1:a_j}',\tilde{u}_{a_{j+1}:a_k})}\frac{\int_{[-1,1]^{|a_j|}} \pi(u_{a_1:a_j}',\tilde{u}_{a_{j+1}:a_k}) \bigotimes_{i\in (a_j)} %\{a_j(l):l\geq 1\}} 
\theta(du_i')}
{\int_{[-1,1]^{|a_j|}} \pi(u_{a_1:a_j}',u_{a_{j+1}:a_k})\bigotimes_{i\in (a_j)} %\{a_j(l):l\geq 1\}} 
\theta(du_i')}\\
&  \leq & \exp\{2\Phi^*\}.
\end{eqnarray*}
Thus, since %as 
$$
M(u,du') = \Big(\prod_{j=1}^{k} \pi(u_{a_j}'|u_{a_1:a_{j-1}}',u_{a_{j+1}:a_{k}}) \Big)\bigotimes_{i=1}^{K} \theta(du_i'),
$$
%with
and
$$
M(\tilde{u},du') = \Big(\prod_{j=1}^{k} \pi(u_{a_j}'|u_{a_1:a_{j-1}}',\tilde{u}_{a_{j+1}:a_{k}}) \Big)\bigotimes_{i=1}^{K} \theta(du_i'),
$$
and the final element in each product is identical, it follows that
$$
M(\tilde{u},du') \geq \exp\{-2\Phi^* (k-1)\} M(u,du').
$$
as was to be proved.
\end{proof}

\section{Numerical Results}\label{sec:numerics}

\subsection{Set-Up}
\label{ssec:numset}

In this section a %we consider a 
1D version of the elliptic PDE problem in \eqref{eq:uniellip} is considered. 
Let $D=[0,1]$ and %consider 2 cases
%\begin{itemize}
%\item[(i)] $u=1$ known, $f$
%\item[(ii)]
%\end{itemize}
consider $f(x)=100x$. For the prior specification of $u$, %we
set $K=2$, $\bar{u}(x)=0.15=const.$, 
$\sigma_1=0.1$, $\sigma_2=0.025$, $\phi_1(x)=\sin(\pi x)$ and $\phi_2(x) = \cos(2\pi x)$.  
The forward problem at resolution level $l$ is solved using a finite element method with 
piecewise linear shape functions on a uniform mesh of width
$h_l=2^{-(l+k)}$, for some starting $k\geq1$ (so that there are at least two grid-blocks in the finest, $l=0$, case).  
%$\cO(h_l^2)$ and hence the global error is 
Thus,  on the $l^{th}$ level the finite-element basis functions are $\{\psi^l_i\}_{i=1}^{2^{l+k}-1}$ defined as (for $x_i=i\cdot 2^{-(l+k)}$) \cite{ciarlet1978finite}:
\[
\quad\quad\quad\quad\quad\quad\quad\quad\quad\quad\quad\quad
\psi_i^{l}(x) = \Bigg \{
\hspace{-50\in}
\begin{split}
%\frac{1}{h} 
(1/h_l)[x - (x_i-h_l) ] &\quad if \quad x\in[x_i-h_l,x_i], \quad\quad\quad\quad\quad\quad
\quad\quad\quad\quad\quad\quad\quad\quad\quad\quad\quad\quad
\quad\quad\quad\quad
%\quad
%\quad
%\quad\quad\quad\quad\quad
\\
%\frac{1}{h} 
(1/h_l)[x_i+h_l -x ] & \quad if \quad x\in [x_i,x_i+h_l]. %\quad j \neq i, 
\quad\quad\quad\quad\quad\quad\quad
\quad\quad\quad\quad\quad\quad\quad\quad\quad\quad\quad
\quad\quad\quad\quad\quad%\quad
%\quad\quad\quad\quad\quad\quad
\end{split}
\]
%linearly interpolating in between.
The functional of interest $g$ is taken as the solution of the forward problem at the midpoint of the domain, that is
$g(u)=p(0.5;u)$.  The observation operator is $\cG(u) = [p(0.25;u),p(0.75;u)]^\top$, and the observational noise 
covariance is taken to be $\Gamma=0.25^2 I$.

To solve the PDE, the ansatz $p_l(x) = \sum_{i=1}^{2^{l+k}-1} p^l_i \psi^l_i(x)$ is plugged into 
\eqref{eq:uniellip}, and projected onto each basis element:
\[
- \Big \langle \nabla \cdot\Big ( \hat{u} \nabla  \sum_{i=1}^{2^{l+k}} p^l_i \psi^l_i(x) \Big), \psi^l_j(x) \Big \rangle 
= \langle f , \psi^l_j \rangle \ ,
\]
resulting in the following linear system:
\[
{\bf A}^l(u) {\bf p}^l = {\bf f}^l,
\]
where we introduce the matrix ${\bf A}^l(u)$ with entries 
$A^l_{ij}(u) = \langle \widehat{u} \nabla \psi^l_i , \nabla \psi^l_j \rangle$, and vectors 
${\bf p}^l, {\bf f}^l$ with entries $p^l_i$ and $f^l_i =  \langle f , \psi^l_i \rangle$, respectively.
Omitting the index $l$, the matrix is sparse and tridiagonal with 
\begin{equation*}
A_{(i-1)i}(u) = A_{i(i-1)}(u)=-(1/h^2) \int_{x_{i-1}}^{x_i} \widehat{u}(x)dx\ ,\quad  
A_{ii}=(1/h^2) \left ( \int_{x_{i-1}}^{x_i} \widehat{u}(x)dx + \int_{x_{i}}^{x_{i+1}} \widehat{u}(x)dx \right)\ , 
\end{equation*}
 and 
zero otherwise.  The elements $f_i$ are computed analogously.  The system can therefore be solved with 
cost
$\cO(2^{l+k})$, corresponding to a computational cost rate of $\gamma=1$.

To get some understanding about the numerics and validate the theory, a number of
results and figures will be generated. 
%We obtain the 
First, the PDE solution is obtained for a reference value of $u$ %(say the one used generate the data)
on a very fine mesh. %{\bf Kody: which mesh exactly??}.
%, say $h=2^{10}$ or larger.  
%We use 
This reference value of $p$ is used to  numerically obtain
the rate $\beta$ in upper bounds of the form $h_l^{\beta}$  for the quantities  in 
\eqref{eq:femconv}, hence also in  \eqref{eq:lincrement}, over increasing $l$.  
Then, %we will optimally allocate 
$N_l$ are optimally allocated using this 
$\beta$ and the $\gamma$ above %(this rate can also be estimated if unsure) 
using the formulae from 
Section \ref{ssec:multilevel_component}.  
Following the error analysis in Section \ref{sec:verify}, once $\beta$ has been decided, 
we have $\alpha=\beta/2$. 
Then, observing the cost/error trend for a range 
of errors $\epsilon$, we expect to observe the appropriate scaling between computational 
cost and mean squared error (e.g.\@ MSE $\propto$ cost$^{-1}$ for MLSMC).

%\textbf{Alex: I find this last paragraph confusing, do we ever need the `ground truth'? XXXX}
%As for the "ground truth" against which to measure the error, the most sensible thing would be to take the high-resolution
%solution above, and also hit it with an enormous number of samples.  All we need is basically one order of magnitude 
%higher accuracy than the smallest error we wish to look at.  If we run into trouble with this, we can always fabricate a
%linear problem by letting $u=\bar{u}$ and $f=a x$, where $a$ is unknown, and say a priori Gaussian, but still solve the 
%problem once for each sample.  This way we have the analytical solution against which to compare.
% 

\subsection{Results}

The following setting is simulated. %At level $l$, $h_l = 2^{-(l + k)}$, $k\ge1$.
The sequence of step-sizes is given by $h_l = 2^{-(l + k)}$, $k = 3$.
The data $\cG(u)$ is simulated with %$u_1 = 0.25$, $u_2 = 0.75$, and 
a given $u_i \sim U[-1,1]$ (i=1,2) and $h = 2^{-20}$. The
observation variance and other algorithmic elements are as stated above. 
We will contrast the accuracy of two algorithms.  
The first is (i) MLSMC as detailed above; the second is 
(ii) plain SMC: the same sequence of distributions 
as MLSMC, but using 
equal number of particles for a given $L$, 
and averaging only the samples at the last level.
For  both MLSMC and SMC algorithms, random walk MCMC kernels
were used (iterated 10 times) with scale parameters falling deterministically 
(the ratio of standard deviation used for  target $\eta_l$ 
versus the one for target $\eta_{l+1}$
is set to $(l+1)/l$).

\subsubsection{Numerical Estimation of Algorithmic Rates}
To numerically estimate the rate $\beta$, 
%first the solution $p(\cdot;u)$ is computed
%with $u$ being the one used for simulating the data and discretisation step 
%$h_L = 2^{-L}$, $L =15$.  Then 
the quantity  
$\|p_l(\cdot;u) - p_{l-1}(\cdot;u)\|_V$ is computed over increasing
levels $l$. Figure~\ref{fig:rate1} shows these computed values plotted against $h_l$ on
base-2 logarithmic scales. A fit of a linear model 
 gives rate $\beta = 1.935$, and a similar experiment 
gives $\alpha=0.993$.
 This is consistent with the rate $\beta=2$ and $\alpha=\beta/2$ expected 
 from the theoretical error analysis in Section \ref{sec:verify}
 (and agrees also with other literature \cite{ciarlet1978finite}). 
An expensive preliminary MLSMC is executed to get some first results over the algorithmic 
variabilty. In this execution the number of particles are set with the recursion 
$N_l = \lceil{2N_{l + 1}}\rceil$ and $N_L =
1000$. The simulations are repeated 100 times.
The estimated variance of $\eta_l^{N_l}(gG_l) / \eta_l^{N_l}(G_l) - \eta_l^{N_l}(g)$, as a
proxy of $V_l$, is %computed. 
plotted in Figure~\ref{fig:var} %plots this variance
%(multiplied by $N_l$) 
against $h_l$ on the same scales as before. The estimate
of the rate now is $\beta = 5.06$. %so it so happens in 
In this case %that 
the numerical estimate is much stronger %more optimistic 
than the theoretical rate used here.  
In fact, under suitable regularity conditions one may theoretically 
obtain the rate $\beta=4$ with a stronger $L^2(D)$ bound on 
$\|p(\cdot;u) - p_{l}(\cdot;u)\|$, which follows from an Aubin-Nitsche duality argument
\cite{ern2004theory}.  However, even this stronger estimate is still beat by the empirical estimate.
Nonetheless, the objective of the present work is to illustrate the theory and not to really
optimize the implementation.  In fact, similar results as presented below are obtained using either 
rate, presumably owing to the fact that $\beta=2>\zeta$, 
which is already the optimal relationship of $\beta$ and $\zeta$ and hence already 
provides the optimal asymptotic behavior of MSE$\propto$cost$^{-1}$.  
In case an optimal $\beta$ induces a change in the relationship between
$\beta$ and $\zeta$, one may expect a %more substantial discrepancy 
change in asymptotic behavior of MSE vs. cost, % from using a sub-optimal rate, 
which justifies such empirical rate estimation.

%gives rate $2\alpha = 2.177$. Note
%that the decay slows down as $h_l$ becomes smaller an smaller. A similar
%pattern can be seen later when we directly estimate the bias. 
%Similarly we compute $\|p_l(\cdot;u) - p_{l-1}(\cdot;u)\|_V$ over
%increasing $l$ and the results is shown in Figure~\ref{fig:rate1}. 

%\begin{figure}\centering
%  \includegraphics{ratel}
%  \caption{Rate estimates.}
%  \label{fig:ratel}
%\end{figure}

% 

\begin{figure}\centering
  \includegraphics{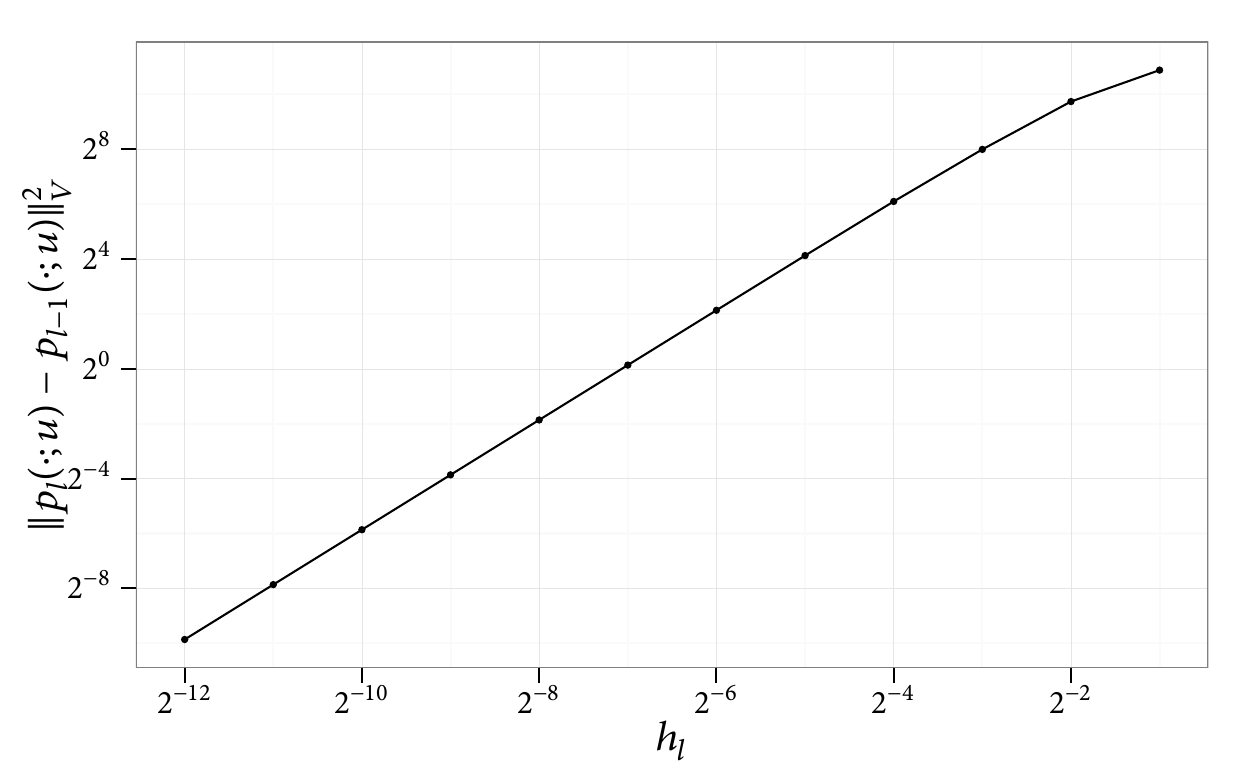}
  \caption{An analyical calculation 
  of $\|p_l(\cdot;u)-p_{l-1}(\cdot;u)\|^{2}_V$,  with $u$ equal to the true value used to generate the data, 
 for various choices 
  of $h_l$.}
  \label{fig:rate1}
\end{figure}

%To verify these bounds, we first compute the true posterior mean for 
%the functional
%$g(U) = p(0.5;U)$ with a MLSMC algorithm that targets $\eta_L$ (\textbf{Alex: $L=?$, $N_l=?$}).  The asymptotic (with respect
%to $N_l$) bias and variance of the estimator are estimated using the results
%from this sampler. 
%
%%The squared bias for single level estimators are plotted against $h_l$ in
%%Figure~\ref{fig:bias}, again on a base 2 logarithm scale. This decreasing of
%%bias slows down gradually as $h_l$ decreases. If we fit a linear model, it
%%gives rate $\alpha = 0.493$.
%%Next, 
% The results for the bias and variance are roughly what the theory predicts.

%\begin{figure}\centering
%  \includegraphics{smc_bias}
%  \caption{Bias estimates.}
%  \label{fig:bias}
%\end{figure}

\begin{figure}\centering
  \includegraphics{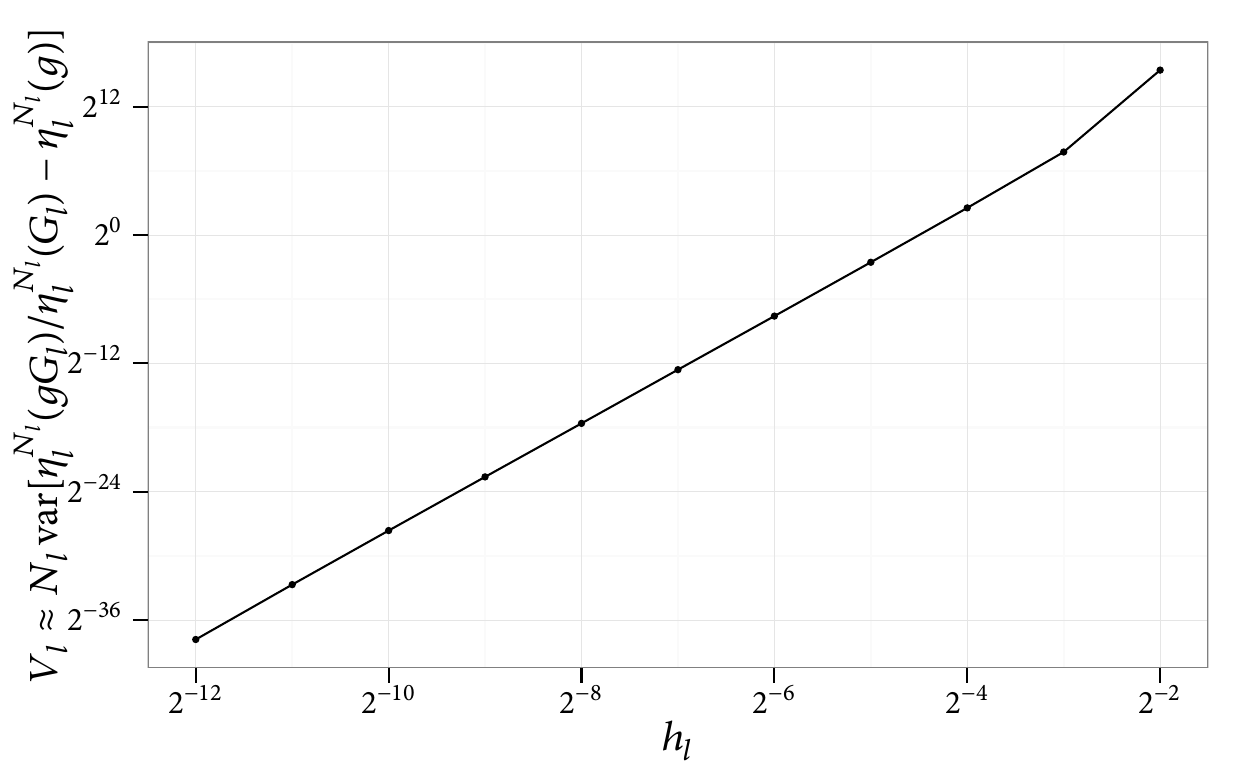}
  \caption{Variance estimates.}
  \label{fig:var}
\end{figure}

\subsubsection{Algorithmic Performance with Diminishing MSE}

Given the choices of $\alpha$ and $\beta$ as above, 
the performance of the MLSMC algorithm is benchmarked by simulating samplers with different
maximum levels $L$. 
The value of $\eta_\infty(g)$ was first 
estimated with the SMC algorithm targeting $\eta_{13}(g)$ ($h^{-16}$), 
with $N_L = 1000$.  This sampler was realized 100 times and 
the average of the estimator is take as the ground truth.
The standard deviation is much smaller than the smallest bias of subsequent simulations.
When updating $L \rightarrow %L'=
L+1$, %as 
the new bias is approximately a factor 
$2^{-\alpha}$ smaller than %times 
the previous one.
%We try to balance 
Therefore the two sources 
of error in \eqref{eq:dec} can be roughly balanced by setting $N_l'=2^{2\alpha}N_l$, 
for $l=0,1,\ldots,L$, and $N_{L+1}' = 2^{-(\beta+\zeta)/2}N'_{L}$. 
% 
%For each $L$ we set $N_L$ first (we are in the situation where $\beta >
%\zeta$). In addition, we set the rest $N_l$ by recursion $N_l = \lceil{2^{(\beta +
%    \zeta) / 2}\rceil N_{l + 1}}$, for rate parameters $\zeta = 1$, $\beta = 2$. 
To check the effectiveness of the MCMC steps employed for dispersing the particles within 
the SMC methods, we show in Figure~\ref{fig:acc}  the average (over the number of particles)
acceptance probability for %all 
each of the $L$ iterations %instances 
when %this 
the MCMC was executed %steps was requested
%(we consider a case when 
(here $L=15$).
The plot indicates reasonable performance of this
particular aspect of the sequential algorithm.

\begin{figure}[!h]\centering
  \includegraphics{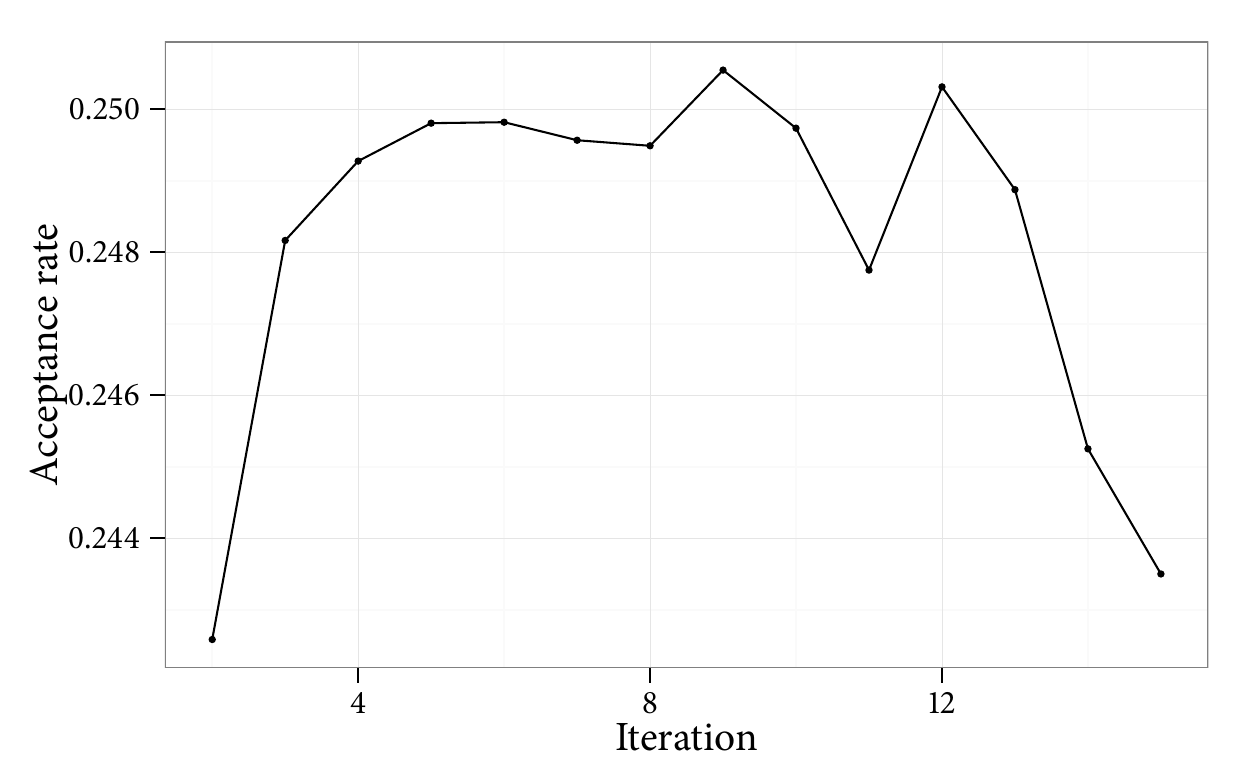}
  \caption{Acceptance rates of MCMC kernel.}
  \label{fig:acc}
\end{figure}

The error-vs-cost plots for SMC and MLSMC are shown in
Figure~\ref{fig:mlsmc}. 
Note that the bullets in the graph correspond to different choices of $L$ 
(ranging from $L=0$ to $L=5$). 
Then, as mentioned earlier, for a given $L$, %for
the single level SMC %we 
uses a fixed number of particles over
the sequence of targets over $l=0,1,\ldots,L$, and %tune 
this number is tuned 
to have approximately the same computational cost as MLSMC with the same $L$.  
The MSE data points are each estimated with $100$ realizations of the given sampler.
The fitted linear model of 
$\log\textrm{MSE}$ against
$\log\text{Cost}$ has a gradient of $-0.6493$ and $-1.029$ for SMC and MLSMC
respectively. %Thus, the second result 
This verifies numerically %experimentally 
the expected asymptotic behavior %connection 
MSE$\propto$cost$^{-1}$ for MLSMC, determined from the theory.  
%The 
Furthermore, the first rate indicates that the single level
SMC performs similarly to the single level vanilla MC with asymptotic behavior MSE$\propto$cost$^{-2/3}$.  
%This latter 
The results clearly establish %es 
the potential improvements of MLSMC versus a standard SMC sampler.
It is remarked that the MLSMC algorithm can be improved in many directions and this is subject to future work.

\begin{figure}\centering
  \includegraphics[scale=1.05]{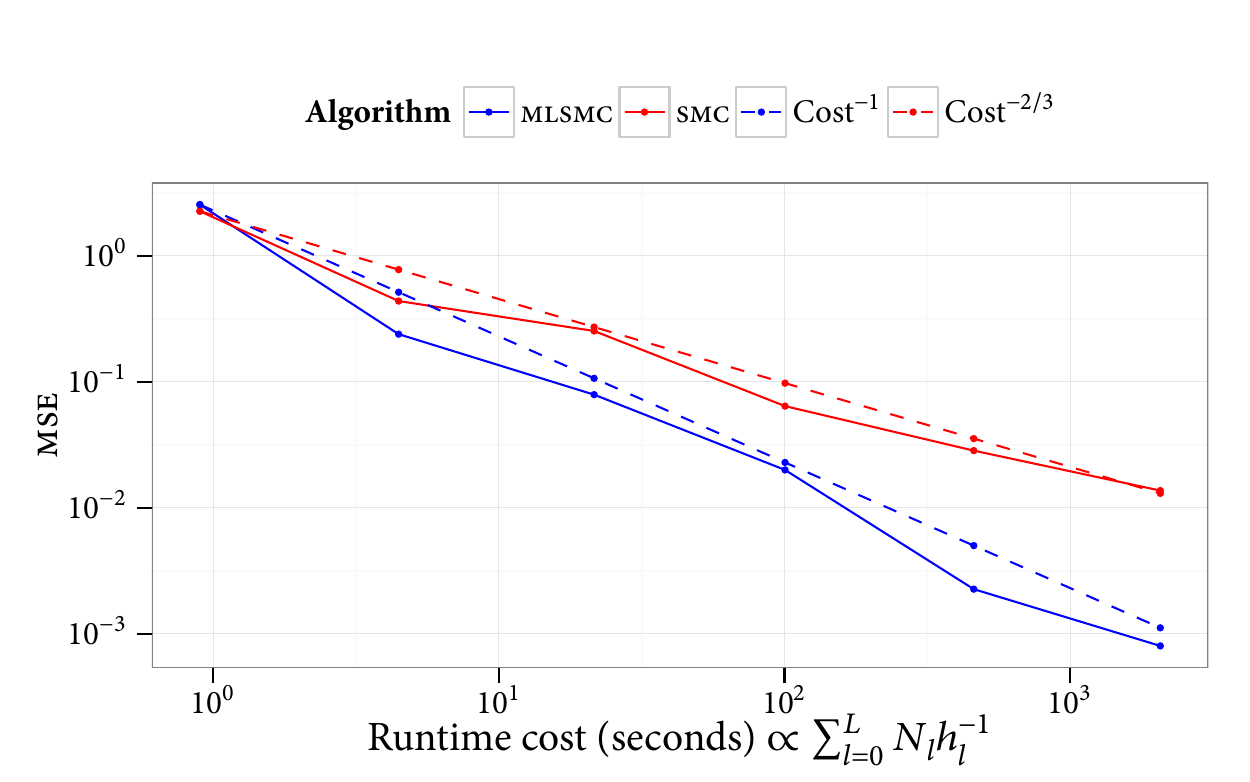}
  \caption{Mean square error against computational cost.}%\mse against cost}
  \label{fig:mlsmc}
\end{figure}

\subsubsection*{Acknowledgements}
AJ, KL \& YZ were supported by an AcRF tier 2 grant: R-155-000-143-112. AJ is affiliated with the Risk Management Institute and the Center for Quantitative Finance at NUS. RT, KL \& AJ were additionally supported by 
King Abdullah University of Science and Technology (KAUST).  AB was supported  by the Leverhulme Trust Prize.

\appendix

\section{Technical Results}

% We will analyze the particle system with $N$, the number of particles fixed per-time step. This is for two reasons:
%(1) The result is of independent interest, (2) The notations are simple. At the end of the Section, we will present results for the case where the
%particle size is deterministically falling, i.e.~$+\infty>N_0\geq N_1\geq\cdots$, which are the results required for this article.

Introduce the following notations. For $\varphi\in\mathcal{B}_b(E)$, $p\geq 0$ and $\eta\in\mathcal{P}(E)$ 
$$
\Phi_p(\eta)(\varphi) = \frac{\eta(G_{p-1}M_{p}(\varphi))}{\eta(G_{p-1})}
%\frac{\eta(G_{p}M_{p+1}(\varphi))}{\eta(G_{p})}
$$
%where  $M_{p+1}(\varphi)(u) = \int_E \varphi(v) M_{p+1}(u,dv)$.
where  $M_{p}(\varphi)(u) = \int_E \varphi(v) M_{p}(u,dv)$.
Define the operator $Q_{p+1}(u,dv) = G_p(u) M_{p+1}(u,dv)$ and denote $Q_{p,n}(\varphi)= Q_{p+1}(\cdots Q_n(\varphi))$ ($0\leq p  \leq n$, $Q_{n,n}$ is the identity operator). %n\leq n$). 
Also set
$$
D_{p,n}(\varphi) = \frac{Q_{p,n}(\varphi-\eta_n(\varphi))}{\eta_p(Q_{p,n}(1))},
$$
$D_{n,n}$ is the identity operator,
%We also 
and define the following %make the definitions
\begin{eqnarray}
V_p^{N_p}(\varphi) & = & \sqrt{N_{p}}\,[\eta_p^{N_{p}}-
\Phi_p(\eta_{p-1}^{N_{p-1}}
)](\varphi) \ ,  \nonumber\\
R_{p+1}^{N_p}(D_{p,n}(\varphi)) & = & \frac{\eta_p^{N_p}(D_{p,n}(\varphi))}{\eta_p^{N_p}(G_p)}[\eta_p(G_p)-\eta_p^{N_p}(G_p)] \ , \label{eq:r_def}
\end{eqnarray}
with the convention that $\Phi_0(\eta_{-1}^{N_{-1}})\equiv \eta_0$.
Working similarly to the derivation of \cite[Eq.~(6.2)]{delm:12}, 
but now with varying number of particles, we have that 
 for any $n\geq 0$
\begin{align*}
[\eta_{n}^{N_n} - \eta_n](\varphi) &= [\eta_n^{N_n} - \Phi_{n}(\eta_{n-1}^{N_{n-1}})](\varphi) +  [\Phi_{n} (\eta_{n-1}^{N_{n-1}}) - 
\eta_n](\varphi) \\[0.2cm]
 &= \frac{V_{n}^{N_n}(\varphi)}{\sqrt{N_{n}}} + \frac{\eta_{n-1}^{N_{n-1}}(G_{n-1}M_{n}(\varphi))}{\eta_{n-1}^{N_{n-1}}(G_{n-1})} - \frac{\eta_{n-1}(G_{n-1}M_{n}(\varphi))}{\eta_{n-1}(G_{n-1})}\\
 &= \frac{V_{n}^{N_n}(\varphi)}{\sqrt{N_n}} + R_n^{N_{n-1}}(D_{n-1,n}(\varphi)) + [\eta_{n-1}^{N_{n-1}} -\eta_{n-1}](D_{n-1,n}(\varphi))
\end{align*}
where notice that $D_{n-2,n-1}D_{n-1,n}=D_{n-2,n}$. Thus, working iteratively  
we have that
\begin{equation}
[\eta_n^{N_n}-\eta_n](\varphi) = \sum_{p=0}^n \frac{V_p^{N_p}(D_{p,n}(\varphi))}{\sqrt{N_p}} +  \sum_{p=0}^{n-1}R_{p+1}^{N_p}(D_{p,n}(\varphi)).
\label{eq:delmoral_decomp}
\end{equation}
Throughout this Section $C$ is a constant whose value may change, but does not depend on any time parameters of the Feynman-Kac formula, nor $(N_0,\dots,N_{L-1})$.

%We n
Now %introduce 
a technical Lemma is introduced, which will contain results that are frequently used in the below calculations.

\begin{lem}\label{lem:tech_lem_imp}
Assume (A\ref{hyp:A}-\ref{hyp:B}). There exist $C<+\infty$, $\kappa\in(0,1)$
such that for any $n\geq p\geq 0$, $q\geq s\geq 0$, $1\leq r <+\infty$
and $\varphi_n,\varphi_q\in\mathcal{B}_b(E)$: 
\begin{enumerate}[i)]
\item  \hspace{0.2cm} $\|D_{p,n}(\varphi_n)\|_{\infty}\leq C\kappa^{n-p}\|\varphi_n\|_{\infty} $.   \label{lem:tech_lem1}
\item \hspace{0.2cm} $\|V_p^{N_p}(D_{p,n}(\varphi_n))\|_r \leq C \kappa^{n-p}\|\varphi_n\|_{\infty}$. 
  \label{lem:tech_lem2}
\item \hspace{0.2cm} $
\|R_{p+1}^{N_p}(D_{p,n}(\varphi_n))\|_r \leq \frac{C \kappa^{n-p}\|\varphi_n\|_{\infty}}{N_p}
$.   \label{lem:tech_lem3}
\item\hspace{0.2cm}
$
\|V_p^{N_p}(D_{p,n}(\varphi_n))V_s^{N_s}(D_{s,q}(\varphi_q))\|_1 \leq C \kappa^{n-p+q-s}\|\varphi_n\|_{\infty}\|\varphi_q\|_{\infty}
$.  \label{lem:tech_lem4}
\item \hspace{0.2cm} $
\|V_p^{N_p}(D_{p,n}(\varphi_n))R_{s+1}^{N_s}(D_{s+1,q}(\varphi_q))\|_1 \leq \frac{C \kappa^{n-p+q-s}\|\varphi_n\|_{\infty}\|\varphi_q\|_{\infty}}{N_s}$.  \label{lem:tech_lem5}
\item \hspace{0.2cm} 
$
\|R_{p+1}^{N_p}(D_{p+1,n}(\varphi_n))R_{s+1}^{N_s}(D_{s+1,q}(\varphi_q))\|_1 \leq \frac{C \kappa^{n-p+q-s}\|\varphi_n\|_{\infty}\|\varphi_q\|_{\infty}}{N_p N_s}
$. \label{lem:tech_lem6}
\end{enumerate}
\end{lem}

\begin{proof}
For \eqref{lem:tech_lem1}. This follows from standard calculations in the analysis of Feynman-Kac formulae; see e.g.\@ the proof of Proposition 2 in \cite{mart:13}.
For \eqref{lem:tech_lem2}.This follows from \cite[Lemma 7.3.3]{delm:04} and 
\eqref{lem:tech_lem1}.
For \eqref{lem:tech_lem3}. Recall \eqref{eq:r_def} and note that $\eta_p(D_{p,n}(\varphi_n))=0$; then on application of Cauchy-Schwarz and assumption (A\ref{hyp:A}) one has that %we have that
$$
\|R_{p+1}^{N_p}(D_{p,n}(\varphi_n))\|_r \leq C\, \|\eta_p^{N_p}(D_{p,n}(\varphi_n))\|_{2r}\cdot \|\eta_p(G_p)-\eta_p^{N_p}(G_p)\|_{2r}\ .
$$
The result follows from \cite[Theorem 7.4.4]{delm:04} and \eqref{lem:tech_lem1}. 
\eqref{lem:tech_lem4} follows from Cauchy-Schwarz and \eqref{lem:tech_lem2}.  \eqref{lem:tech_lem5} follows from Cauchy-Schwarz, \eqref{lem:tech_lem2} and \eqref{lem:tech_lem3}. 
\eqref{lem:tech_lem6} follows from Cauchy-Schwarz and \eqref{lem:tech_lem3}.
\end{proof}

%We d
Recall equations \eqref{eq:ay} and \eqref{eq:aybar}, and define the following terms,
\begin{align*}
\mathcal{V}_{n}(\varphi, {N}) = \sum_{p=0}^n 
\frac{V_p^{N_p}(D_{p,n}(\varphi))}{\sqrt{N_p}} \ ; \quad 
\mathcal{R}_{n}(\varphi, {N}) = 
\sum_{p=0}^{n-1}R_{p+1}^{N_p}(D_{p,n}(\varphi))\ .
\end{align*}
Here we use a slight abuse of notation for $N$, representing $(N_0,\dots,N_n)$ (or $(N_0,\dots,N_{n-1})$).

\begin{lem}\label{lem:tech_lem_res1.5}
Assume (A\ref{hyp:A}-\ref{hyp:B}). There exist a $C<+\infty$, $\kappa\in(0,1)$ such that for any $n> q \geq 0$ and $\varphi_n,\varphi_q,g\in\mathcal{B}_b(E)$, $\|g\|_{\infty}=1$:
\end{lem}
\begin{itemize}
\item[i)]  \hspace{0.2cm}  $\big|\, \mathbb{E}\,[\,A_q(g,N_q)\mathcal{V}_{n}(\varphi_n,{N})\mathcal{V}_{q}(\varphi_q,{N})\,]\,\big|   \leq \frac{C\|\varphi_n\|_{\infty}\|\varphi_q\|_{\infty} \kappa^{n-q}}{N_q}$.   
\item[ii)]  \hspace{0.2cm} $ \big|\, \mathbb{E}\,[\,A_q(g,N_q)\mathcal{V}_{n}(\varphi_n,{N})\mathcal{R}_{q}(\varphi_q,{N})\,]\,\big|  \leq \frac{C\|\varphi_n\|_{\infty}\|\varphi_q\|_{\infty} \kappa^{n-q}}{N_q^{3/2}} $.  
\item[iii)]  \hspace{0.2cm} $ \big|\,\mathbb{E}\,[\,A_q(g,N_q)\mathcal{V}_{q}(\varphi_q,{N})\mathcal{R}_{n}(\varphi_n,{N})\,]\, \big| \leq \frac{C\|\varphi_n\|_{\infty}\|\varphi_q\|_{\infty}}{\sqrt{N_q}N_n} $.  
\item[iv)]  \hspace{0.2cm} $ \big|\, \mathbb{E}\,[\,A_q(g,N_q)\mathcal{R}_{q}(\varphi_q,{N})\mathcal{R}_{n}(\varphi_n,{N})\,]\,\big|  \leq \frac{C\|\varphi_n\|_{\infty}\|\varphi_q\|_{\infty}}{N_n N_q} $.
\end{itemize}
\begin{proof}
Set $\mathcal{F}^{{N}}_{q }$ as the $\sigma$ algebra generated by particle system up to time $q$.

i)\,We start by noting that $\mathbb{E}\,[\,V_p^{N_p}(D_{p,n}(\varphi_n))\mid\mathcal{F}^{{N}}_{q }\,]=0$ for any $q< p \leq n$, so that:
\begin{equation*}
\mathbb{E}\,[\,A_q(g,N_q)\mathcal{V}_{n}(\varphi_n,N)\mathcal{V}_{q}(\varphi_q,N)\,] = \sum_{0\le p,s\le q} \mathbb{E}\,\big[\,\tfrac{1}{\sqrt{N_p N_s}}\,A_q(g,N_q)V_p^{N_p}(D_{p,n}(\varphi_n))V_s^{N_s}(D_{s,q}(\varphi_q))\,\big]\ .
\end{equation*}
As $|A_q(g,N_q)|\leq 1$, one can use Lemma \ref{lem:tech_lem_imp} \eqref{lem:tech_lem4}, to obtain the upper bound 
\begin{align*}
 \sum_{0\le p,s\le q} \mathbb{E}\,\big|\,\tfrac{1}{\sqrt{N_p N_s}}\,A_q(g,N_q)&V_p^{N_p}(D_{p,n}(\varphi_n))V_s^{N_s}(D_{s,q}(\varphi_q))\,\big| \leq  C\|\varphi_n\|_{\infty}\|\varphi_q\|_{\infty}\sum_{0\le p,s\le q}  \tfrac{\kappa^{n-p+q-s}}{\sqrt{N_p N_s}}\\
&\leq \tfrac{C\|\varphi_n\|_{\infty}\|\varphi_q\|_{\infty} \kappa^{n-q}}{N_q}\ .
\end{align*}

\noindent ii)\,\, Again using  $\mathbb{E}\,[\,V_p^{N_p}(D_{p,n}(\varphi_n))\mid \mathcal{F}^N_{q}\,]=0$ for $q< p \leq n$, we have
\begin{equation*}
 \mathbb{E}\,[\,A_q(g,N_q)\mathcal{V}_{n}(\varphi_n,N)\mathcal{R}_{q}(\varphi_q,N)\,]
= \sum_{0\le p,s\le q} \mathbb{E}\,\big[\,\tfrac{1}{\sqrt{N_p}}\,A_q(g,N_q)V_p^{N_p}(D_{p,n}(\varphi_n))R_{s+1}^{N_s}(D_{s,q}(\varphi_q))\,\big]\ .
\end{equation*}
By $|A_q(g,N_q)|\leq 1$ and Lemma \ref{lem:tech_lem_imp}\eqref{lem:tech_lem5} we have the upper-bound
\begin{align*}
 \sum_{0\le p,s\le q} \mathbb{E}\,\big|\,\tfrac{1}{\sqrt{N_p}}A_q(g,N)&V_p^{N_p}(D_{p,n}(\varphi_n))R_{s+1}^{N_s}(D_{s,q}(\varphi_q))\,\big| 
\leq   C\|\varphi_n\|_{\infty}\|\varphi_q\|_{\infty} \sum_{0\le p,s\le q}  \tfrac{\kappa^{n-p+q-s}}{\sqrt{N_p}\,{N_s} }\\
&\leq \tfrac{C\|\varphi_n\|_{\infty}\|\varphi_q\|_{\infty} \kappa^{n-q}}{N_q^{3/2}}\ .
\end{align*}

\noindent iii)\,\, 
%$$
%\frac{1}{\sqrt{N}}\sum_{p=0}^{n-1} \sum_{s=0}^q\mathbb{E}\Big[A_q(g,N)R_{p+1}^N(D_{p,n}(\varphi_n))V_s^N(D_{s,q}(\varphi_q))\Big].
%$$
By $|A_q(g,N_q)|\leq 1$ and Lemma \ref{lem:tech_lem_imp}\eqref{lem:tech_lem5}, we have the upper bound
\begin{align*}
\big|\, \mathbb{E}\,[\,A_q(g,N_q)&\mathcal{V}_{q}(\varphi_q,N)\mathcal{R}_{n}(\varphi_n,N)\,]\,\big|  \le  
\sum_{p=0}^{n-1} \sum_{s=0}^q\mathbb{E}\,\big|\,\tfrac{1}{\sqrt{N_s}}A_q(g,N_q)R_{p+1}^N(D_{p,n}(\varphi_n))V_s^N(D_{s,q}(\varphi_q))\,\big|\\ &\leq  C\|\varphi_n\|_{\infty}\|\varphi_q\|_{\infty} 
\sum_{p=0}^{n-1}\sum_{s=0}^{q} \tfrac{ \kappa^{n-p+q-s} }{\sqrt{N_s}N_p}
\leq \tfrac{C\|\varphi_n\|_{\infty}\|\varphi_q\|_{\infty}}{\sqrt{N_q}N_n}\ .
\end{align*}

\noindent iv)\,\,  Working as in (iii), by 
 $|A_q(g,N_q)|\leq 1$ and Lemma \ref{lem:tech_lem_imp}\eqref{lem:tech_lem6}, we have 
\begin{align*}
\big|\, \mathbb{E}\,[\,A_q(g,N_q)&\mathcal{R}_{q}(\varphi_q,N)\mathcal{R}_{n}(\varphi_n,N)\,]\,\big| \le 
\sum_{p=0}^{n-1}\sum_{s=0}^{q-1}\mathbb{E}\,\Big|\,A_q(g,N_q)R_{p+1}^N(D_{p,n}(\varphi_n))R_{s+1}^N(D_{s,q}(\varphi_q))\,\Big| \\ &\leq   C\|\varphi_n\|_{\infty}\|\varphi_q\|_{\infty} 
\sum_{p=0}^{n-1}\sum_{s=0}^{q} \tfrac{\kappa^{n-p+q-s}}{N_p N_s}
\leq \tfrac{C\|\varphi_n\|_{\infty}\|\varphi_q\|_{\infty}}{N_n N_q}\ .
\end{align*}
\end{proof}

\begin{lem}\label{lem:tech_lem_res2}
Assume (A\ref{hyp:A}-\ref{hyp:B}). There exists a $C<+\infty$, such that for any $n\geq 0$, $1\leq r <+\infty$ and $\varphi\in\mathcal{B}_b(E)$:
$$
\|\overline{A}_n(\varphi,N_n)\|_r \leq \tfrac{C\|\varphi\|_{\infty}}{\sqrt{N_n}}.
$$
\end{lem}

\begin{proof}
The result is standard, but we give the proof for completeness. We have
$$
\|\overline{A}_n(\varphi,N_n)\|_r = \big\|\,\tfrac{\eta_n^{N_n}(\varphi G_n)-\eta_n(\varphi G_n)}{\eta_n^{N_n}(G_n)} + \tfrac{\eta_n(\varphi G_n)}{\eta_n^{N_n}(G_n)\eta_n(G_n)}[\eta_n(G_n)-\eta_n^{N_n}(G_n)]\, \big\|_r\ .
$$
Application of Minkowski, (A\ref{hyp:A}) and \cite[Theorem 7.4.4]{delm:04} complete the proof.
\end{proof}

\begin{prop}\label{prop:prop_corr_bd3}
Assume (A\ref{hyp:A}-\ref{hyp:B}). There exist a $C<+\infty$, $\kappa\in(0,1)$ such that for any $n> q \geq 0$ and $\varphi_n,\varphi_q,f,g\in\mathcal{B}_b(E)$, $\|g\|_{\infty} = \|f\|_{\infty} =1$:
\begin{equation*}
\Big|\,\mathbb{E}\,\big[\,A_n(g,N_n)A_q(f,N_q)(\eta_n^{N_n}-\eta_n)(\varphi_n)(\eta_q^{N_q}-\eta_q)(\varphi_q)\,\big]\,\Big| \leq C \|\varphi_n\|_{\infty}\|\varphi_q\|_{\infty} \big(\tfrac{\kappa^{n-q}}{N_q} + \tfrac{1}{N_q^{1/2}N_n}\big)
% + \tfrac{1}{N^{2}}+ \tfrac{1}{N^{5/2}}\ .
\end{equation*}
\end{prop}

\begin{proof}
From the definition of $A_n$, $\overline{A}_n$ we have that
\begin{align*}
\mathbb{E}\,\big[\,A_n&(g,N_n)A_q(f,N_q)(\eta_n^{N_n}-\eta_n)(\varphi_n)(\eta_q^{N_q}-\eta_q)(\varphi_q)\,\big] =\\ 
 &\Delta_{1,q,n}(f,g,\varphi_{q},\varphi_{n}, N_q,N_n) + 
 \tfrac{\eta_n(g G_n)}{\eta_n(G_n)}\,\Delta_{2,q,n}(f,\varphi_{q},\varphi_{n}, N_q,N_n)
\end{align*}
where we have defined 
\begin{align*}
\Delta_{1,q,n}(f,g,\varphi_{q},\varphi_{n}, N_q, N_n) &= 
\mathbb{E}\,\big[\,\overline{A}_n(g,N_n)A_q(f,N_q)(\eta_n^{N_n}-\eta_n)(\varphi_n)(\eta_q^{N_q}-\eta_q)(\varphi_q)\,\big]\ ,\\
\Delta_{2,q,n}(f,\varphi_{q},\varphi_{n}, N_q,N_n) &= \mathbb{E}\,\big[\,A_q(f,N_q)(\eta_n^{N_n}-\eta_n)(\varphi_n)(\eta_q^{N_q}-\eta_q)(\varphi_q)\,\big]\ . 
\end{align*}

By Lemma \ref{lem:tech_lem_res1.5} and the fact that $\frac{\eta_n(|g| G_n)}{\eta_n(G_q)}\leq 1$, we have that
\begin{equation*}
\big|\, \tfrac{\eta_n(g G_n)}{\eta_n(G_n)}\,\Delta_{2,q,n}(f,\varphi_{q},\varphi_{n}, N_q,N_n)\,\big| \leq C \|\varphi_n\|_{\infty}\|\varphi_q\|_{\infty} \big(\tfrac{\kappa^{n-q}}{N_q} + \tfrac{\kappa^{n-q}}{N_q^{3/2}} + \tfrac{1}{\sqrt{N_q}N_n} + \tfrac{1}{N_n N_q}\big)\ .
\end{equation*}
Thus we concentrate on $\Delta_{1,q,n}(f,g,\varphi_{n},\varphi_{q}, N_q,N_n)$. We have via \eqref{eq:delmoral_decomp}:
\begin{align*}
\Delta_{1,q,n}&(f,g,\varphi_{n},\varphi_{q}, N_q, N_n)
 = \\
&=\mathbb{E}\,\big[\,\overline{A}_n(g,N_n)A_q(f,N_q)(\mathcal{V}_n(\varphi_n,N) +
\mathcal{R}_n(\varphi_n,N))(\mathcal{V}_q(\varphi_q,N) +
\mathcal{R}_q(\varphi_q,N))\,\big]\ .
\end{align*}
We will deal with each of the 4 terms on the R.H.S.~separately.

We start with $\mathbb{E}\,[\,\overline{A}_n(g,N_n)A_q(f,N_q)\mathcal{V}_n(\varphi_n,N)\mathcal{V}_q(\varphi_q,N)\,]$ and work as follows,
\begin{align*}
\Big|\,\mathbb{E}\,&[\,\overline{A}_n(g,N_n)A_q(f,N_q)\mathcal{V}_n(\varphi_n,N)\mathcal{V}_q(\varphi_q,N)\,]\,\Big| \\
&=\Big|\,\sum_{p=0}^{n} \sum_{s=0}^q \mathbb{E}\,\big[\,\tfrac{1}{\sqrt{N_p N_s}}\overline{A}_n(g,N_n)A_{q}(f,N_q)V_p^{N_p}(D_{p,n}(\varphi_n))V_s^{N_s}
(D_{s,q}(\varphi_q))\,\big]\,\Big|
 \\
&\le  
\sum_{p=0}^{n} \sum_{s=0}^{q-1}\tfrac{1}{\sqrt{N_p N_s}}\,\|\overline{A}_n(g,N)\|_3\,
\|V_p^{N_p}(D_{p,n}(\varphi_n))\|_{3}\,
\|V_s^{N_s}(D_{s,q}(\varphi_q))\|_3\ \\
&\le   C\|\varphi_n\|_{\infty}\|\varphi_q\|_{\infty} \tfrac{1}{\sqrt{N_n}} \sum_{p=0}^{n} \sum_{s=0}^q \tfrac{\kappa^{n-p+q-s}}{\sqrt{N_p N_s}}\\ 
&\leq \tfrac{C\|\varphi_n\|_{\infty}\|\varphi_q\|_{\infty}}{N_n\sqrt{N_q}}\ .
\end{align*}
where for the third line we have used $|A_{q}(g,N_q)|\leq 1$
and two applications of H\"older's inequality;
for the forth line we have used Lemma \ref{lem:tech_lem_imp}(\ref{lem:tech_lem2}) and Lemma \ref{lem:tech_lem_res2}.

Using very similar calculations one can obtain 
the upper bounds,
\begin{eqnarray*}
\Big|\,\mathbb{E}\,[\,\overline{A}_n(g,N_n)A_q(f,N_q)\mathcal{V}_n(\varphi_n,N)\mathcal{R}_q(\varphi_q,N)\,]\,\Big| & \le &  \tfrac{C\|\varphi_n\|_{\infty}\|\varphi_q\|_{\infty}}{N_n N_q}\ , \\
\Big|\,\mathbb{E}\,[\,\overline{A}_n(g,N_n)A_q(f,N_q)\mathcal{R}_n(\varphi_n,N)\mathcal{V}_q(\varphi_q,N)\,]\,\Big| & \le &
\tfrac{C\|\varphi_n\|_{\infty}\|\varphi_q\|_{\infty}}{N_n^{3/2}N_q^{1/2}}\ , \\
\Big|\,\mathbb{E}\,[\,\overline{A}_n(g,N)A_q(f,N)\mathcal{R}_n(\varphi_n)\mathcal{R}_q(\varphi_q)\,]\,\Big| & \le & \tfrac{C\|\varphi_n\|_{\infty}\|\varphi_q\|_{\infty}}{N_n^{3/2}N_q}.
\end{eqnarray*}
The proof is now complete.
%\textbf{Term 2}. We have to deal with the expression
%$$
%\frac{1}{\sqrt{N}}\sum_{p=0}^n \sum_{s=0}^q \mathbb{E}\Big[\overline{A}_n(g,N)A_{q}(g,N)V_p^N(D_{p,n}(\varphi_n))R_{s+1}^N(D_{s,q}(\varphi_q))\Big].
%$$
%Using $A_{q}(g,N)\leq 1$ and the same calculations as for Term 1 in the proof Proposition \ref{prop:prop_corr_bd2} we have:
%$$
%\frac{1}{\sqrt{N}}\sum_{p=0}^n \sum_{s=0}^q \mathbb{E}\Big[\overline{A}_n(g,N)A_{q}(g,N)V_p^N(D_{p,n}(\varphi_n))R_{s+1}^N(D_{s,q}(\varphi_q))\Big] \leq
%\frac{C\|\varphi_n\|_{\infty}\|\varphi_q\|_{\infty}}{N^{2}}.
%$$
%
%\textbf{Term 3 \& 4}. These can be treated using $A_{q}(g,N)\leq 1$ and the same calculations as for the corresponding term in the proof of Proposition \ref{prop:prop_corr_bd2}.
%That is, we have
%$$
%\mathbb{E}\Big[\overline{A}_n(g,N)A_q(g,N)\Big(\sum_{p=0}^{n-1}R_{p+1}^N(D_{p,n}(\varphi_n))\Big)\Big( \sum_{s=0}^q \frac{V_s^N(D_{s,q}(\varphi_q))}{\sqrt{N}}\Big)\Big] \leq \frac{C\|\varphi_n\|_{\infty}\|\varphi_q\|_{\infty}}{N^{2}}
%$$
%and
%$$
%\mathbb{E}\Big[\overline{A}_n(g,N)A_q(g,N)\Big(\sum_{p=0}^{n-1}R_{p+1}^N(D_{p,n}(\varphi_n))\Big)\Big(\sum_{s=0}^{q-1}R_{s+1}^N(D_{s,q}(\varphi_q))\Big)\Big] \leq \frac{C\|\varphi_n\|_{\infty}\|\varphi_q\|_{\infty}}{N^{5/2}}.
%$$
%
%The proof is completing on summing the upper-bounds on the terms 1-4 and using simple calculations which are omitted.
\end{proof}

\end{document}